\documentclass[preprint]{elsarticle}

\usepackage{hyperref}
\usepackage{amssymb}
\usepackage{amsmath}
\usepackage{lineno}
\usepackage{tikz}

\def\computationproblem#1#2#3{
  \begin{center}
  \begin{tabular}{rp{0.8\textwidth}}
  {\sc Problem:\enspace}&#1\\
  {\sc Input:\enspace}&#2\\
  {\sc Question:\enspace}&#3\\
  \end{tabular}
  \end{center}
}

\newtheorem{theorem}{Theorem}
\newtheorem{lemma}[theorem]{Lemma}
\newtheorem{proposition}[theorem]{Proposition}
\newtheorem{observation}[theorem]{Observation}
\newtheorem{problem}[theorem]{Problem}
\newtheorem{corollary}[theorem]{Corollary}
\newdefinition{definition}[theorem]{Definition}
\newproof{proof}{Proof}

\begin{document}

\begin{frontmatter}

\title{Computational Complexity of Covering Multigraphs with Semi-edges: Small Cases\tnoteref{funding}
}

\tnotetext[funding]{An extended abstract of this paper previously appeared at MFCS 2021 as~\cite{n:BFHJK-MFCS}. The second, fourth, and fifth author were supported by research grant GAČR 20-15576S of the Czech Science Foundation. The third author was supported by research grant GAČR 20-04567S of the Czech Science Foundation. The fourth author was partially supported by GAUK 370122. In the late stages of preparing the manuscript, the first author was partially supported by the ANR project GRALMECO (ANR-21-CE48-0004), by the French government IDEX-ISITE initiative 16-IDEX-0001 (CAP 20-25), and by
European Union (ERC, POCOCOP, 101071674). Views and opinions expressed are however those of the author(s) only and do
not necessarily reflect those of the European Union or the European Research Council Executive Agency. Neither the European
Union nor the granting authority can be held responsible for them.}

\author[5,4]{Jan Bok\corref{cor1}}
\cortext[cor1]{Corresponding author}
\ead{jan.bok@matfyz.cuni.cz}

\author[3]{Jiří Fiala}
\ead{fiala@kam.mff.cuni.cz}

\author[2]{Petr Hliněný}
\ead{hlineny@fi.muni.cz}

\author[5,3]{Nikola Jedličková}
\ead{jedlickova@kam.mff.cuni.cz}

\author[3]{Jan Kratochvíl}
\ead{honza@kam.mff.cuni.cz}

\address[1]{Computer Science Institute, Faculty of Mathematics and Physics, Charles University, Prague, Czech Republic}
\address[2]{Faculty of Informatics, Masaryk University, Brno, Czech Republic}
\address[3]{Department of Applied Mathematics, Faculty of Mathematics and Physics, Charles University, Prague, Czech Republic}
\address[4]{Universit\'e Clermont Auvergne, CNRS, Clermont Auvergne INP, Mines Saint-\'Etienne, LIMOS, 63000 Clermont-Ferrand, France}
\address[5]{Department of Algebra, Faculty of Mathematics and Physics, Charles University, Prague, Czech Republic}


\begin{abstract}
We initiate the study of computational complexity of graph coverings, aka
locally bijective graph homomorphisms, for {\em graphs with semi-edges}. The
notion of graph covering is a discretization of coverings between surfaces
or topological spaces, a notion well known and deeply studied in classical
topology. Graph covers have found applications in discrete mathematics for
constructing highly symmetric graphs, and in computer science in the theory of
local computations. In 1991, Abello, Fellows, and Stillwell asked for a classification of
the computational complexity of deciding if an input graph covers a fixed
target graph, in the ordinary setting (of graphs with only edges). 
Although many general results are known, the full classification is still open.
In spite of that, we propose to study the more general case of covering
graphs composed of normal edges (including multiedges and loops) and so-called semi-edges.
Semi-edges are becoming increasingly popular in modern topological graph theory, 
as well as in mathematical physics.
They also naturally occur in the local computation setting, 
since they are lifted to matchings in the covering graph.
We show that the presence of semi-edges makes the covering problem considerably
harder; e.g., it is no longer sufficient to specify the vertex mapping induced
by the covering, but one necessarily has to deal with the edge mapping as well. 
We show some solvable cases and, in particular, completely characterize the complexity of
the already very nontrivial problem of covering one- and two-vertex (multi)graphs with semi-edges. 
Our NP-hardness results are proven for simple input graphs, and in the case 
of regular two-vertex target graphs, even for bipartite ones. 
We remark that our new characterization results also strengthen previously known
results for covering graphs without semi-edges,
and they in turn apply to an infinite class of simple target graphs
with at most two vertices of degree more than two.
Some of the results are moreover proven in a more general setting (e.g., finding $k$-tuples of pairwise
disjoint perfect matchings in regular graphs).  
\end{abstract}

\end{frontmatter}

\section{Introduction}\label{s:intro}

\subsection{Graph coverings and complexity}

The notion of a {\em graph covering} is a discretization of coverings between surfaces or topological spaces, a notion well known and deeply studied in classical topology. 
A covering projection from one topological space to another one is a many-to-one mapping such that every point of the covering space and its image in the covered one have open neighborhoods which are mapped homeomorphically to each other by the projection. This is a motivation for the introduction of the notion of a graph covering projection as a many-to-one mapping between two graphs such that, informally speaking, every vertex of the covering graph and its image ``look alike'' with respect to the edges that are incident with them. More precisely, if the image $f(x)$ of a vertex $x$ of the covering graph is adjacent to another vertex $y$, $x$ has exactly one neighbor $z$ which is mapped to $y$ by the covering projection $f$.  In other words, if a graph $G$ covers a graph $H$ via a covering projection $f$, the covering graph $G$ is a perfect maze resembling $H$. In the sense that if an agent keeps moving through the vertices of $G$ and when visiting a vertex, say $x$, the agent can see the image $f(x)$ of $x$ and for each edge incident with $x$ the image of the other end-vertex of this edge, (s)he cannot distinguish if (s)he is moving through $G$ or through $H$. 

Graph coverings have found many applications. Primarily as a tool for construction of highly symmetric graphs~\cite{k:Biggs74,n:Djokovic74,n:Gardiner74,n:GT77}, or for embedding complete graphs in surfaces of higher genus~\cite{k:Ringel74}.

Graph coverings attracted attention of computer scientists as well. Angluin~\cite{n:Angluin80} exploited graph covers when introducing models of local computations, namely by showing that a graph and its cover cannot be distinguished by local computations. Later, Litovsky, M\'{e}tivier, and Zielonka~\cite{n:LMZ93} proved that planar graphs and series-parallel graphs cannot be recognized by local computations, and Courcelle and M\'{e}tivier~\cite{n:CM94} showed that in fact no nontrivial minor-closed class of graphs can. In both of these results, graph coverings were used as the main tool, as well as in more recent papers~\cite{n:ChMZ06,n:ChalopinP11}. Here, the authors presented a model for distributed computations and addressed the algorithmic complexity of problems associated with such a model. To this end, they used the existing results on NP-completeness of the covering problem to provide their hardness results. In~\cite{packing_bipartite}, the authors study a close relation of packing bipartite graphs to a special variant of graph coverings called \emph{pseudo-coverings}.

Another connection to algorithmic theory comes through the notions of the {\em degree partition} and the {\em degree refinement matrix} of a graph. These notions were introduced by Corneil~\cite{Corneil68,n:CorneilG70} in hope of solving the graph isomorphism problem efficiently. It can be easily seen that a graph and all of its covers have the same degree refinement matrix. 
Motivated by this observation, Angluin and Gardiner~\cite{n:AG81} proved
that any two finite regular graphs of the same valency have a finite common
cover, and conjectured the same for every two finite graphs
with the same degree refinement matrix, which was later proved by
Leighton~\cite{n:Leighton82}.

The stress on finiteness of the common cover is natural. For every matrix,
there exists a universal cover, an infinite tree, that covers all graphs
with this degree refinement matrix.  Trees are planar graphs, and this
inspired an at first sight innocent question of which graphs allow a finite
planar cover.  Negami observed that projective planar graphs do (in fact,
their double planar covers characterize their projective embedding), and
conjectured that these two classes actually coincide~\cite{n:Negami88}. 
Despite a serious effort of numerous authors, the
problem is still open, although the scope for possible failure of Negami's
conjecture has been significantly reduced
\cite{n:Archdeacon02,n:Hlineny98,n:HT04}.

A natural computational complexity question is how difficult is to decide, given two graphs, if one covers the other one. 
This question is obviously at least as difficult as the graph isomorphism problem (consider two given graphs on the same number of vertices). It was proven NP-complete by Bodlaender~\cite{n:Bodlaender89} (in the case of both graphs being part of the input). Abello, Fellows, and Stillwell~\cite{n:AFS91} initiated the study of the computational complexity of the $H$-cover problem for a fixed target graph $H$ by showing that deciding if an input graph covers 
the dumbbell graph $W(0,1,1,1,0)$ (in our notation from Section~\ref{s:twovertex}) is NP-complete (note that the dumbbell graph has loops, and they also allowed the input graph to contain loops). 
Furthermore, they asked for a complete characterization of the computational
complexity, depending on the parameter graphs $H$.  Such a line of research
was picked by Kratochv\'{\i}l, Proskurowski, and Telle, who
completely characterized the complexity for simple target graphs with at
most 6 vertices~\cite{n:KPT98}, and then noted that in order to fully
characterize the complexity of the $H$-cover problem for simple target
graphs, it is sufficient (but also necessary) to classify it for mixed
colored multigraphs with minimum degree at least three~\cite{n:KPT97a}.  The
latter result gives a hope for a more concise description of the
characterization, but is also in line with the original motivation
of covers
from topological graph theory, where loops and multiedges are widely considered.

The complexity of covering $2$-vertex multigraphs was fully characterized in~\cite{n:KPT97a}, the characterization for 3-vertex undirected multigraphs can be found in~\cite{n:KratochvilTT16}. The most general NP-hardness result known so far is the hardness of covering simple regular graphs of valency at least three~\cite{n:Fiala00b,n:KPT97}. More recently, the authors of~\cite{n:BilkaJKTV11} proved that covering several concrete small graphs (including the complete graphs $K_4, K_5$ and $K_6$) remains NP-hard for planar inputs. This shows that planarity does not help in graph covering problems in general, yet the conjecture that the \textsc{$H$-Cover} problem restricted to planar inputs is at least as difficult as for general inputs, provided $H$ itself has a finite planar cover, remains still open. Planar graphs have also been considered by the authors of~\cite{n:FialaKKN14} who showed that for planar input graphs, \textsc{$H$-RegularCover} is in FPT when parameterized by $H$.
This is in fact the first and only paper on the complexity of regular covers, i.e., covering projections determined by a regular action of a group of automorphisms on the covering graph.

Graph coverings were also extensively studied under a unifying umbrella of
\emph{locally constrained homomorphisms}.  In these relaxations,
homomorphisms can be either locally injective or locally surjective and not
necessarily locally bijective.  The computational complexity of locally
surjective homomorphisms has been classified completely, with respect to the
fixed target graph~\cite{n:FP05}.  
Though the complete classification of the
complexity of locally injective homomorphisms is still out of sight, it has
been proved for its list variant~\cite{n:FK06}.  The problem is also
interesting for its applied motivation -- a locally injective homomorphism
into the complement of a path of length $k$ corresponds to an
$L(2,1)$-labeling of span $k$, an intensively studied notion stemming from
the theory of frequency assignment.  Further generalizations include the
notion of $H(p,q)$-coloring, a homomorphism into a fixed target graph $H$
with additional rules on the neighborhoods of the
vertices~\cite{n:FHKT03,n:KT00}.  
To find more about locally injective
homomorphisms, see e.g.~\cite{n:ChaplickFHPT15,n:FK08,n:MacGillivrayS10}.  
For every fixed graph $H$, the
existence of a locally injective homomorphism to $H$ is provably at least as
hard as the $H$-cover problem.  
In this sense our hardness results extend the state of the art also for the 
problem of existence of locally injective homomorphisms.

\subsection{Graphs with semi-edges}

The topological motivation quickly leads to considering graphs with multiple edges and loops (edges that connect a vertex to itself). When defining covering projections for such graphs (often referred to as \emph{multigraphs}), one should pay attention to mapping the edges as well as mapping the vertices, these two components of the projection must be compatible (incidence-preserving). It can be, however, easily seen that in order to specify a covering projection, it is sufficient to specify the mapping of vertices, as long as it fulfills certain degree constraints (cf.~Theorem~\ref{p:dob3}). Further generalization leads to considering \emph{semi-edges}, i.e., objects that are incident with one vertex only. The difference from loops is that a semi-edge adds 1 to the degree of its incident vertex, while a loop contributes 2. In the topological motivation, a loop provides two ways how the agent can move from a vertex to itself in one step, while a semi-edge opens just one such possibility. Though this interpretation may look artificial, it becomes natural in the covering graphs. A semi-edge can be the image of an edge whose both end-vertices are mapped to the vertex incident with the semi-edge. And thus in the covering graph, each of the end-vertices has exactly one edge leading to a vertex mapped to the same image (by the way, this is also the reason why loops are not allowed to be mapped to semi-edges). 

The notion of {\em semi-edges} has been introduced
in the modern topological graph theory and it is becoming more and more frequently used
(the terminology has not yet stabilized; semi-edges are often called half-edges, and sometimes fins).
Mednykh and Nedela recently wrote a monograph~\cite{nedela_mednykh} in which
they summarize and survey the ambitions and efforts behind generalizing the
notion of graph coverings to the graphs with semi-edges. 
This generalization, as the authors pinpoint, is not artificial
as such graphs emerge ``in the situation of taking quotients of simple graphs by groups of
automorphisms which are semiregular on vertices and darts (arcs) and which
may fix edges''.
As the authors put it: ``A problem arises when one wants
to consider quotients of such graphs (graphs embedded to surfaces) by an
involution fixing an edge $e$ but transposing the two incident vertices. 
The edge $e$ is halved and mapped to a semi-edge --- an edge with one free
end.'' This direction of research proved to be very fruitful and provided
many applications and generalizations to various parts of algebraic graph
theory.  For example, Malni{\v{c}}, Maru\v{s}i\v{c}, and Poto\v{c}nik~\cite{n:MalnivcMP04} considered
semi-edges during their study of abelian covers and as they write ``...in
order to have a broader range of applications we allow graphs to have
semi-edges.'' To highlight a few other contributions, the reader is invited
to consult~\cite{n:NedelaS96,n:MalnicNS00}, the surveys
\cite{kwak2007graphs} and (aforementioned)~\cite{nedela_mednykh}, and
finally for more recent results the series of papers~\cite{n:FialaKKN14,arxiv1609.03013,n:FialaKKN18}.  
It is also worth noting that the concept of graphs with semi-edges was introduced independently and
naturally in mathematical physics by Getzler and
Karpanov~\cite{getzler1998modular}.

From the graph-theoretical point of view, semi-edges allow an elegant and unifying description of several crucial graph-theoretical notions. For instance, as seen in Section~\ref{s:onevertex}, a simple cubic graph is 3-edge-colorable if and only if it covers the one-vertex graph with three semi-edges, and a simple cubic graph contains a perfect matching if and only if it covers the one-vertex graph with one loop and one semi-edge. 

In the view of the theory of local computations, semi-edges and their covers prove very natural, too, and it is even surprising that they have not been considered before in the context. If a computer network is constructed as a cover of a small template, the preimages of normal edges in the covering projection are matchings completely connecting nodes of two types (the end-vertices of the covered edge). Preimages of loops are disjoint cycles with nodes of the same type. And preimages of semi-edges are matchings on vertices of the same type. The role of semi-edges was spotted by Woodhouse~\cite{woodhouse_2021} and Shepherd~\cite{n:Shep22} who have generalized the fundamental theorem of Leighton~\cite{n:Leighton82} on finite common covers of graphs with the same degree refinement matrix to graphs with semi-edges.

The goal of the conference version~\cite{n:BFHJK-MFCS} of this paper was to initiate the study of the computational complexity of covering graphs with semi-edges, and that paper surely opened the door in the intended direction. The authors of~\cite{n:BFJKS21-FCT} studied the computational complexity of covering disconnected graphs and showed that, under an adequate yet natural definition of covers of disconnected graphs, the existence of a covering projection onto a disconnected graph is polynomial-time decidable if an only if it is the case for every connected component of the target graph. In~\cite{DBLP:conf/iwoca/BokFJKR22} the authors explicitly formulated the Strong Dichotomy Conjecture stating that for every fixed target (multi)graph with semi-edges allowed, the covering problem is either polynomial-time solvable for arbitrary input graphs, or NP-complete for simple graphs as input. Most recently,   a complete characterization of the computational complexity of covering colored mixed (multi)graphs with semi-edges whose every class of degree partition contains at most two vertices was presented in~\cite{WG23}. That provides a common generalization of the observation for simple graphs from~\cite{n:KPT97a}, and for the case of regular 2-vertex graphs in our Theorem~\ref{t:2-vertex}. It should be noted, though, that the general NP-hardness reduction in~\cite{WG23} strongly benefits from our Proposition~\ref{prop:covering}, namely from the fact that several of the covering problems remain NP-hard for simple bipartite input graphs. The current paper is an enhanced version of~\cite{n:BFHJK-MFCS} and contains all proofs in full.

\subsection{Formal definitions}

In this subsection we formally define what we call {\em graphs}. A graph has
a set of vertices and a set of edges (also referred to as links).  As it is
standard in topological graph theory, we
automatically allow multiple edges and loops.  Every ordinary edge is
connecting two vertices, every loop is incident with only one vertex.  On
top of these, we also allow {\em semi-edges}.  Each semi-edge is also
incident with only one vertex.  The difference between loops and semi-edges
is that a loop contributes two to the degree of its vertex, while a
semi-edge only one.  A very elegant description of
ordinary edges, loops and semi-edges through the concept of {\em darts} is used in more algebraic-based papers on covers.
The following formal definition is a reformulation of
the one given in~\cite{nedela_mednykh}.

\begin{definition}\label{def:graph-dart} 
A \emph{graph} is a triple $(D,V,\Lambda)$, where $D$ is a set of \emph{darts},
and $V$ and $\Lambda$ are each a partition of $D$ into disjoint sets.
Moreover, all sets in $\Lambda$ have size one or two.
\end{definition}

With this definition, the {\em vertices} of a graph $(D,V,\Lambda)$ are the
sets of $V$ (note that empty sets correspond to isolated vertices, and since
we are interested in covers of connected graphs by connected
ones, we assume that all sets of $V$ are nonempty).  The sets of $\Lambda$
are referred to as {\em links}, and they are of three types -- loops
(2-element sets with both darts from the same set of $V$), (ordinary) edges
(2-element sets intersecting two different sets of $V$), and semi-edges
(1-element sets).  After this explanation it should be clear that this
definition is equivalent to a definition of multigraphs which is standard in
the graph theory community:

\begin{definition}\label{def:graph-old} 
A {\em graph} is an ordered triple $(V,\Lambda,\iota)$, for $\Lambda=E\cup L\cup S$, where $\iota$ is the {\em incidence mapping} $\iota:\Lambda\longrightarrow V\cup{V\choose{2}}$ such that $\iota(e)\in V$ for all $e\in L\cup S$ and $\iota(e)\in {V\choose{2}}$ for all $s\in E$. 
\end{definition}

\begin{figure}
\centering
\includegraphics[width=0.9\textwidth]{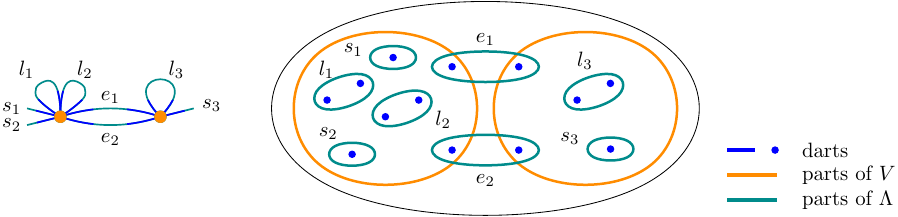}
\caption{An example of a graph presented in a usual graph-theoretical way (left) and using the  dart-based
Definition~\ref{def:graph-dart} (right).}
\label{fig:different-graph-definitions}
\end{figure}

For a comparison of Definitions~\ref{def:graph-dart} and~\ref{def:graph-old}, see Figure~\ref{fig:different-graph-definitions}. 
Since we consider multiple edges of the same type incident with the same vertex (or with the same pair of vertices), the edges are given by their names and the incidence mapping $\iota$ expresses which vertex (or vertices) `belong' to a particular edge. The degree of a vertex is then defined as follows.

\begin{definition}\label{def:degree}
For a graph $G=(V,\Lambda=E\cup L\cup S,\iota)$, the degree of a vertex $u\in V$ is defined as
$$\mbox{deg}_G(u)=p_{S}(u)+p_{E}(u)+2p_L(u),$$
where $p_{S}(u)$ (\,$p_{L}(u)$) is the number of semi-edges $e\in S$
(of loops $e\in L$) such that $\iota(e)=u$, 
and $p_{E}(u)$ is the number of ordinary edges $e\in E$ such that $u\in \iota(e)$.
\end{definition}

We call a graph $G$ {\em simple} if $p_{S}(u)=p_L(u)=0$ for every vertex $u\in V(G)$ (the graph has no loops or semi-edges) and $\iota(e)\neq\iota(e')$ for every two distinct $e,e'\in E$ (the graph has no multiple (ordinary) edges). We call $G$ {\em semi-simple} if $p_{S}(u)\le 1$ and $p_L(u)=0$ for every vertex $u\in V(G)$ and $\iota(e)\neq\iota(e')$ for every two distinct $e,e'\in E$.

Note that in the language of Definition~\ref{def:graph-dart}, the degree of a vertex $v\in V$ is simply $|v|$. We say that a graph is \emph{regular} if all its vertices have the same degree and we say it is \emph{$k$-regular} if all its vertices have the same degree $k$.

In this language, the main object of our study, a \emph{graph cover} (or equivalently a \emph{covering projection}), is defined as follows.

\begin{definition}\label{def:graph-dart-cover}
We say that a graph $G=(D_G,V_G,\Lambda_G)$ \emph{covers} a connected graph $H=(D_H,V_H,\Lambda_H)$ (denoted as $G\longrightarrow H$) if there exists a map 
$f\colon D_G\to D_H$ such that:
\begin{itemize}
\item For every $u\in V_G$, there is a $u'\in V_H$ such that the restriction of $f$ onto $u$ is a bijection between $u$ and $u'$.
\item For every $e\in \Lambda_G$, there is an $e'\in \Lambda_H$ such that $f(e)=e'$.
\end{itemize}
The map $f$ is called \emph{graph cover} (or \emph{covering projection}).
\end{definition}

One must appreciate how compact and elegant this definition is after translating it into the language of Definition~\ref{def:graph-old} in Proposition~\ref{prop:covering},
which otherwise is the definition of (multi)graph covering in the standard language of Definition~\ref{def:graph-old}. 

The fact that a loop contributes 2 to the degree of its vertex may seem strange at first sight, but becomes natural when graphs are considered embedded to surfaces, and is absolutely obvious when we look at the definition of a covering projection (for the sake of exactness, the definition is somewhat technical, we promise to be less formal in the rest of the paper).

\begin{proposition}\label{prop:covering}
A graph $G$ covers a connected graph $H$ if and only if $G$ allows 
a pair of mappings $f_V:V(G)\longrightarrow V(H)$ and $f_\Lambda:\Lambda(G)\longrightarrow \Lambda(H)$ such that
\begin{enumerate}
\item $f_\Lambda(e)\in L(H)$ for every $e\in L(G)$
 and $f_\Lambda(e)\in S(H)$ for every $e\in S(G)$,
\item $\iota(f_\Lambda(e))=f_V(\iota(e))$ for every $e\in L(G)\cup S(G)$,
\item for every link $e\in \Lambda(G)$ such that $f_\Lambda(e)\in S(H)\cup L(H)$ and $\iota(e)=\{u,v\}$, we have $\iota(f_\Lambda(e))=f_V(u)=f_V(v)$,
\item for every link $e\in \Lambda(G)$ such that $f_\Lambda(e)\in E(H)$ and $\iota(e)=\{u,v\}$ (note that it must be $f_V(u)\neq f_V(v)$), we have $\iota(f_\Lambda(e))=\{f_V(u),f_V(v)\}$,
\item for every loop $e\in L(H)$, $f^{-1}(e)$ is a disjoint union of loops and cycles spanning all vertices $u\in V(G)$ such that $f_V(u)=\iota(e)$,
\item for every semi-edge $e\in S(H)$,  $f^{-1}(e)$ is a disjoint union of edges and semi-edges spanning all vertices $u\in V(G)$ such that $f_V(u)=\iota(e)$, and
\item for every edge $e\in E(H)$,  $f^{-1}(e)$ is a disjoint union of edges (i.e., a matching) spanning all vertices $u\in V(G)$ such that $f_V(u)\in\iota(e)$.
\end{enumerate}
\end{proposition}  

\begin{figure}
\begin{center}
\scalebox{0.9}{\includegraphics[page=1]{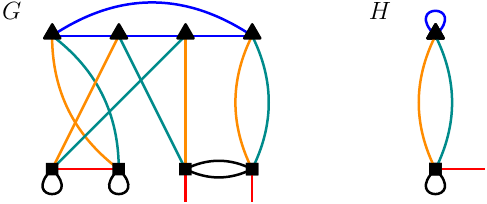}}
\end{center}
\caption{An example of a covering. The vertex mapping of the covering from $G$ to $H$ is determined by the shape of the vertices, the edge mapping by the colors of the edges.}\label{fig:examplecover}
\end{figure}

See an example of a covering projection in Figure~\ref{fig:examplecover}. 
Conditions 1--4 express the fact that $f_V$ and $f_E$ commute with $\iota$, i.e., that $f$ is a homomorphism from $G$ to $H$. 
Conditions 5--7 express that this homomorphism is locally bijective:
\begin{itemize}
  \item for every ordinary edge $e$ incident with $f_V(u)$ in $H$, there is exactly one ordinary edge of $G$ which is incident with $u$ and mapped to $e$ by $f_E$;
  \item for every semi-edge $e$ incident to $f_V(u)$ in $H$, there is exactly one semi-edge, or exactly one ordinary edge (but not both) in $G$ incident with $u$ and mapped to $e$ by $f_E$;
  \item and for every loop $e$ incident with $f_V(u)$ in $H$, there is exactly one loop or exactly two ordinary edges (but not both) of $G$ which are incident with $u$ and mapped to $e$ by $f_E$.
\end{itemize}

Even though the aforementioned definitions of graphs and graph  covers through darts are compact and elegant, in the rest of the paper we shall work with the
standard definition of graphs and the equivalent description of graph covers given by Proposition~\ref{prop:covering}, 
because they are better suited for describing our reductions and understanding the illustrative figures.

It is clear that a covering projection (more precisely, its vertex mapping) preserves degrees. One may ask when (or if) a degree preserving vertex mapping can be extended to a covering projection. An obvious necessary condition is described by the following definition.

\begin{definition} A vertex mapping $f_V:V(G)\longrightarrow V(H)$ between graphs $G$ and $H$ is called {\em degree-obedient} if 
\begin{enumerate}
\item for any two distinct vertices $u,v\in V(H)$ and any vertex $x\in f_V^{-1}(u)$, the number of ordinary edges $e$ of $H$ such that $\iota(e)=\{u,v\}$ equals the number of ordinary edges of $G$ with one end-vertex $x$ and the other one in $f_V^{-1}(v)$, and
\item for every vertex $u\in V(H)$ and any vertex $x\in f_V^{-1}(u)$, the value $p_{S(H)}(u)+2p_{L(H)}(u)$ equals 
 $p_{S(G)}(x)+2p_{L(G)}(x)+r$, where $r$ is the number of edges of $G$ with one end-vertex $x$ and the other one from $f^{-1}_V(u)\setminus\{x\}$,
\item  for every vertex $u\in V(H)$ and any vertex $x\in f^{-1}_V(u)$, $p_{S(G)}(x)\le p_{S(H)}(u)$. 
\end{enumerate}    
\end{definition} 

Finally, let us recall that 
the \emph{product} $G \times H$ of simple graphs $G$ and $H$ is defined as the graph with the vertex set being the Cartesian product $V(G) \times V(H)$ and with
vertices $(u,v)$ and $(u',v')$ being adjacent in $G \times H$ if and only if
$u$ is adjacent to $u'$ in $G$, and
$v$ is adjacent to $v'$ in $H$. We will use the following generalization in several of our constructions.

\begin{definition}
Let $G=(V,\Lambda=E\cup L\cup S,\iota)$ be a (multi)graph. Its product $G\times K_2$ with a single edge is the graph with vertex set $V_1\cup V_2$, where $V_i=\{u_i:u\in V\}, i=1,2$, and edge set $E'=\{e_{12},e_{21}:e\in E\}\cup \{\ell_1,\ell_2:\ell\in L\}\cup \{s':s\in S\}$ with the incidence function defined as follows:
\begin{itemize}
    \item if $\iota(e)=\{u,v\}$ for and ordinary edge $e\in E$ of $G$, then $\iota(e_{12})=\{u_1,v_2\}$ and $\iota(e_{21})=\{u_2,v_1\}$ (or vice versa),
    \item if $\iota(\ell)=u$ for a loop $\ell$ of $G$, then $\iota(\ell_1)=\iota(\ell_2)=\{u_1,u_2\}$, and
    \item if $\iota(s)=u$ for a semi-edge $s\in S$ of $G$, then 
    $\iota(s')=\{u_1,u_2\}$.
\end{itemize}
\end{definition}

Note that $G\times K_2$ is bipartite and has no loops and no semi-edges.

\subsection{Overview of our results}

The first major difference between graphs with and without semi-edges is that for target graphs without semi-edges, every degree-obedient vertex mapping to it can be extended to a covering. This is not true anymore when semi-edges are allowed (consider a one-vertex graph with three semi-edges, every 3-regular graph allows a degree-obedient mapping onto it, but only the 3-edge-colorable ones are covering it). 
In Section~\ref{s:semi-edges} we show that the situation is not as bad if the source graph is bipartite. In Theorem~\ref{p:dob3} we prove that if the source graph is bipartite and has no semi-edges, then every degree-obedient vertex mapping can be extended to a covering, while if semi-edges are allowed in the bipartite source graph, it can at least be decided in polynomial time if a degree-obedient mapping can be extended to a covering. 

All other results concern the complexity of the following decision problem

\computationproblem
{\sc $H$-Cover}
{A graph $G$.}
{Does $G$ cover $H$?}   

In order to present our results in the strongest possible form, we aim at proving the hardness results for restricted classes of input graphs, while the polynomial ones for the most general inputs. In particular, we only allow simple graphs as inputs when we prove NP-hardness, and on the other hand, we allow loops, multiple edges as well as semi-edges when we present polynomial-time algorithms. 

The first NP-hardness result is proven in Theorem~\ref{t:simplegraphs}, namely that covering semi-simple regular graphs of valency at least 3 is NP-hard even for simple bipartite input graphs. In Sections~\ref{s:onevertex} and~\ref{s:twovertex} we give a complete classification of the computational complexity of covering graphs with
one and two vertices. This extends the main result of~\cite{n:KPT97a} to graphs with semi-edges. Moreover, we strengthen the hardness results of~\cite{n:KPT97a} considerably by showing that all NP-hard cases of covering regular two-vertex graphs (even those without semi-edges) remain NP-hard for simple {\em bipartite} input graphs. Note that through the reduction from~\cite{n:KPT97b}, our results on the complexity of covering one- or two-vertex graphs provide characterization results on infinitely many simple graphs which contain at most two vertices of degrees greater than 2. 

All considered computational problems are clearly in the class NP, and thus we only concentrate on the NP-hardness proofs in the NP-completeness results. We restrict our attention to connected target graphs, in which case it suffices to consider only connected input graphs. In this case every cover is a $k$-fold cover for some $k$, which means that the preimage of every vertex has the same size.

\section{The impact of semi-edges}\label{s:semi-edges}

In this section we demonstrate the huge difference between covering graphs with and without semi-edges. 
First, we discuss the necessity of specifying the edge mapping in a covering projection.
In other words, we discuss when a degree mapping can always be extended to a covering,
and when this question can be decided efficiently.
The following proposition follows straightforwardly from the definitions.

\begin{proposition}\label{p:dob1}
For every graph covering projection between two graphs, the vertex mapping induced by this projection is degree-obedient.
\end{proposition}

\begin{proposition} \label{p:dob2}
If $H$ has no semi-edges, then for any graph $G$, any degree-obedient mapping from the vertex set of $G$ onto the vertex set of $H$ can be extended to a graph covering projection of $G$ to $H$. 
\end{proposition}

\begin{proof}
For simple graphs $G$, this is proved already in~\cite{n:KPT97a}. If multiple edges and loops are allowed, we use a similar approach. The key point is that Petersen's theorem~\cite{petersen1891} about 2-factorization of regular graphs of even valence is true for multigraphs without semi-edges as well, and the same holds true for the K\"onig-Hall theorem~\cite{k:LP86} on 1-factorization of regular bipartite multigraphs.
\end{proof}

As we will see soon, the presence of semi-edges changes the situation a lot. Even for simple graphs, degree-obedient vertex mappings to a graph with semi-edges may not extend to a graph covering projection, and the possibility of such an extension may even be NP-complete. 

\begin{observation}
Let $F(3,0)$ be the graph with one vertex and three semi-edges pending on this vertex. Then a simple graph covers $F(3,0)$ if and only if it is 3-regular and 3-edge-colorable. Testing 3-edge-colorability is well known to be NP-hard for simple graphs.  
\end{observation}

However, if the input graph is bipartite, the situation gets much easier.   

\begin{theorem}\label{p:dob3} 
If a graph $G$ is bipartite, then for any graph $H$, it can be decided in polynomial time whether a degree-obedient mapping from the vertex set of $G$ onto the vertex set of $H$ can be extended to a graph covering projection of $G$ to $H$. In particular, if $G$ has no semi-edges and is bipartite, then every degree-obedient mapping from the vertex set of $G$ onto the vertex set of $H$ can be extended to a graph covering projection of $G$ to $H$.
\end{theorem}

\begin{proof}
Let $G$ be a bipartite graph and let $f_V:V(G)\longrightarrow V(H)$ be a degree-obedient mapping from the vertex set of $G$ to a vertex set of $H$. We seek an edge mapping $f_E:E(G)\longrightarrow E(H)$ such that $(f_V,f_E)$ is a covering projection of $G$ to $H$. For every edge or semi-edge $s$ of $G$, its image under $f_E$ is restricted to be chosen from edges with corresponding end-vertices: if $s$ is a semi-edge on vertex $u$, $f_E(s)$ must be a semi-edge on $f_V(u)$, and if $s$ is an edge with end-vertices $u$ and $v$ (a loop, when $u=v$), $f_E(s)$ must be an edge with end-vertices $f_V(u)$ and $f_V(v)$ (a loop or a semi-edge, if $f_V(u)=f_V(v)$\,). 

Consider two distinct vertices $x\neq y\in V(H)$, and let they be connected by $k$ edges $e_1,e_2,\ldots, e_k$ in $H$. The  bipartite subgraph $\widetilde{G_{x,y}}$ of $G$ with  classes of bipartition $f_V^{-1}(x)$ and  $f_V^{-1}(y)$ and edges of $G$ with end-points in different classes is $k$-regular. By the K\"onig-Hall theorem, it is $k$-edge colorable. If $\varphi: E(\widetilde{G_{x,y}})\longrightarrow \{1,2,\ldots,k\}$ is such a coloring, then $f_E: E(\widetilde{G_{x,y}})\longrightarrow \{e_1,e_2,\ldots,e_k\}$ defined by $f_E(h)=e_{\varphi(h)}$ is a covering projection onto the set of parallel edges between $x$ and $y$ in $H$.

The situation is more complex for loops and semi-edges of $H$. Consider a vertex $x\in V(H)$ and the subgraph $\widetilde{G_x}$ of $G$ induced by $f_V^{-1}(x)$. If $x$ has $b$ semi-edges and $c$ loops in $H$, $\widetilde{G_x}$ is $(b+2c)$-regular. Let $s(u)$ be the number of semi-edges of $G$ incident with $u$, and set $g(u)=b-s(u)$. In a covering projection, for every $u\in f_V^{-1}(x)$, exactly $g(u)$ of edges incident with $u$ must map onto semi-edges of $H$ incident with $x$. Hence a covering projection on the edges of $\widetilde{G_x}$ exists only if $\widetilde{G_x}$ has a $g$-factor for the above defined function $g$. This can be decided in polynomial time (e.g., by network flow algorithms, since $\widetilde{G_x}$ is a bipartite graph, but even for general graphs the existence of a $g$-factor can be reduced to the maximum matching problem, cf.~\cite{EdmondsJ01}). If such a $g$-factor exists,  it is $b$-edge-colorable (here and only here we use the assumption that $G$ is bipartite), and such an edge-coloring defines a mapping $f_E$ from the edges of the $g$-factor onto the semi-edges of $H$ incident with $x$. For every vertex $u\in f_V^{-1}(x)$, $g(u)$ edges of $G$ incident with $u$ are mapped onto $g(u)$ distinct semi-edges incident with $x$ in $H$, and $b-g(u)=s(u)$ semi-edges remain available as images of the $s(u)$ semi-edges incident with $u$ in $G$. What remains is to define $f_E$ for the so far unmapped edges of $\widetilde{G_x}$. But these form a $2c$-regular graph which covers $c$ loops on $x$ in $H$ (a consequence of Petersen's theorem, or the K\"onig-Hall theorem since $G$ is bipartite and hence the edges of a $2c$-regular bipartite graph can be partitioned into $2c$ perfect matchings and these matchings can be paired into $c$ disjoint spanning cycles, each covering one loop).

If $\widetilde{G_x}$ has no semi-edges, then it is bipartite $(b+2c)$-regular and as such it always has a $b$-factor. Hence for a bipartite graph without semi-edges, a degree-obedient vertex mapping can always be extended to a graph covering projection.
\end{proof}

Now we prove the first general hardness result, namely that covering semi-simple regular graphs is always NP-complete
(this is the case when every vertex of the target graph is incident with at most one semi-edge, and the graph has no multiple edges nor loops). 
See Figure~\ref{fig:examples-semi-simple} for examples of semi-simple graphs $H$ defining such hard cases.

\begin{figure}[ht]
\begin{center}
\scalebox{1.0}{\includegraphics{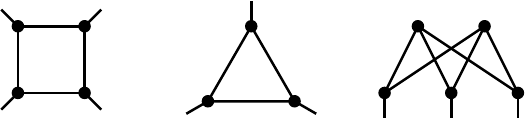}}
\end{center}
\caption{Examples of small semi-simple graphs which define NP-complete covering problems.}\label{fig:examples-semi-simple}
\end{figure}

\begin{theorem}\label{t:simplegraphs}
Let $H$ be a non-empty semi-simple $k$-regular graph, with $k\ge 3$. Then the {\sc $H$-Cover} problem is NP-complete even for simple bipartite input graphs. 
\end{theorem}

\begin{proof}
Consider $H'=H\times K_2$. This graph is simple, $k$-regular and bipartite, hence the {\sc $H'$-Cover} problem is NP-complete by~\cite{n:KPT97}. Given an input $k$-regular graph $G$, it is easy to see that $G$ covers $H'$ if and only it is bipartite and covers $H$. Since bipartiteness can be checked in polynomial time, the claim follows.
\end{proof}

\section{One-vertex target graphs}\label{s:onevertex}

We start the section by proving a slightly more general hardness result, which may be of interest on its own. In particular, it implies that for every $d\ge 3$, it is NP-complete to decide if a simple $d$-regular graph contains an even 2-factor, i.e., a spanning 2-regular subgraph whose every cycle has even length. It is worth noting that this is the same as checking if those graphs have oddness equal to zero. The \emph{oddness} (defined in literature often just for bridgeless cubic simple graphs) is the minimum number of odd cycles in a 2-factor of a given graph. We refer to~\cite{haggkvist2005double,huck2001cycle,mavcajova2014sparsely,mazzuoccolo2017nowhere} for further information regarding this parameter.

\begin{theorem}\label{t-onevertex-general}
For every $k\ge 2$ and every $d\ge k+1$, it is NP-complete to decide if a  simple $d$-regular graph contains $k$ pairwise disjoint perfect matchings.
\end{theorem}

\begin{proof}
The complement of the union of $k$ pairwise disjoint perfect matchings in a $(k+1)$-regular graph is a perfect matching as well, and thus a $(k+1)$-regular graph contains $k$ pairwise disjoint perfect matchings if and only if it is $(k+1)$-edge colorable. Hence for $d=k+1$, the claim follows from the NP-completeness of $d$-edge colorability of $d$-regular graphs which has been proven by Leven and Galil~\cite{n:LevenGalil83}.

Let $d\ge k+2$. We prove the claim by a reduction from $(k+1)$-edge colorability of $(k+1)$-regular graphs (using~\cite{n:LevenGalil83} again). Fix a graph $H$ with one vertex, say $x$, of degree $d-2$ and all other vertices having degrees $d$, and such that $H$ contains $d-2$ pairwise disjoint perfect matchings (such a graph can be easily constructed, see the end of the proof). 
Given a $(k+1)$-regular graph $G$ whose $(k+1)$-edge colorability is questioned,
we construct a graph $G'$ as follows: The graph $G'$ contains two disjoint copies $G_1,G_2$ of $G$
such that the two clones of each vertex $u$ of $G$ in $G_1$ and $G_2$ are connected together by $d-k-1$ paths of length $2$.
Moreover, the middle vertices in each of those paths play the role of the vertex $x$ in a copy of $H$
(each copy of $H$ is private to its path). See Figure~\ref{fig:one-vertex}. Formally, 
$$V(G')=V(G_1) \cup V(G_2) \cup \bigcup_{u\in V(G)}\bigcup_{i=1}^{d-k-1}V(H_{u,i}), \textrm{ and}$$
$$E(G')=E(G_1)\cup E(G_2) \cup \bigcup_{u\in V(G)}\bigcup_{i=1}^{d-k-1}(E(H_{u,i}) \cup \{u_1x_{u,i}, u_2x_{u,i}\},
$$
where 
$$V(G_j)=\{u_j:u\in V(G)\} \mbox{ and }E(G_j)=\{u_jv_j:uv\in E(G)\}$$ 
for $j=1,2$, and 
$$V(H_{u,i})=\{y_{u,i}:y\in V(H)\} \mbox{ and }
E(H_{u,i})=\{y_{u,i}z_{u,i}:yz\in E(H)\}$$ 
for $u\in V(G)$ and $i=1,2,\ldots,d-k-1$. 

\begin{figure}
\begin{center}
\scalebox{1.15}{\includegraphics{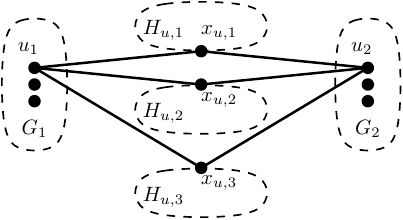}}
\end{center}
\caption{An illustration for the construction of the graph $G'$ in the proof of Theorem~\ref{t-onevertex-general}.}\label{fig:one-vertex}
\end{figure}

We claim that $G'$ has $k$ pairwise disjoint perfect matchings if and only if $\chi'(G)=k+1$ (where $\chi'(G)$ denotes the \emph{chromatic index of graph $G$}, i.e. the smallest number of colors needed to color its edges). In one direction, if $G$ is $k$-edge colorable, then for each $j=1,2$, the graph $G_j$ has $k$ pairwise disjoint perfect matchings, say  $M_h^j, h=1,2,\ldots, k$. By the assumption on $H$, each $H_{u,i}$ has $k\le d-2$ pairwise disjoint matchings, 
say $M_h^{u,i}, h=1,2,\ldots,k$, 
for all $u\in V(G)$ and $i=1,2,\ldots, d-k-1$. Then 
$$M_h=M_h^1\cup M_h^2\cup \bigcup_{u\in V(G)}\bigcup_{i=1}^{d-k-1}M_h^{u,i},$$ 
for $h=1,2,\ldots,k$,
are $k$ pairwise disjoint perfect matchings in $G'$.

For the opposite implication, note that no perfect matching of $G'$ contains any of the edges $u_jx_{u,i}$, $u\in V(G), i=1,2,\ldots,d-k-1, j=1,2$, because each $H_{u,i}$ has
an even number of vertices and each $x_{u,i}$ is an articulation in $G'$. So, for every perfect matching $M$ in $G'$, $M\cap E(G_1)$ is a perfect matching in $G$. Thus if $M_h, h=1,2,\ldots,k$ are pairwise disjoint perfect matchings in $G'$, then  $\{uv\in E(G): u_1v_1\in M_h\}, h=1,2,\ldots,k$ are $k$ pairwise disjoint perfect matchings in $G$, and hence $\chi'(G)=k+1$.

To complete the proof, let us show an explicit construction of the auxiliary graph $H$. Fix an odd number $t\ge d+1$. It is well known that the complete graph $K_{t+1}$ is $t$-edge colorable, i.e., its edge set $E(K_{t+1})$ can be partitioned into $t$ perfect matchings, say $M_1, M_2, \ldots, M_t$ (an elegant proof of this otherwise folklore result can be found in~\cite{k:Soifer08}). Choose vertices $x,y,z$ so that $xy\in M_1, xz\in M_2$, and assume without loss of generality that $yz\in M_t$. Define the graph $H$ as follows
$$V(H)=V(K_{t+1})$$
$$E(H)=(\bigcup_{i=1}^d M_i \setminus \{xy,xz\})\cup \{yz\}.$$
Then $\mbox{deg}_Hx=d-2$ and $\mbox{deg}_Hu=d$ for all $u\in V(H)\setminus \{x\}$. Moreover, $H$ has $d-2$ pairwise disjoint perfect matchings $M_3,M_4,\ldots,M_d$. 
\end{proof}

Now we are ready to prove a dichotomy theorem on the complexity of covering one-vertex graphs. Let us denote by $F(b,c)$ the one-vertex graph with $b$ semi-edges and $c$ loops.

Before proving the next theorem, let us pinpoint that matchings here can consist of both edges and semi-edges. The generalisation is straightforward, putting semi-edges into matchings can be done by a simple preprocessing, and the notion of perfectness naturally applies to such ``generalised matchings''.

\begin{theorem}\label{t-onevertex}
The {\sc $F(b,c)$-Cover} problem  is polynomial-time solvable if $b\le 1$, or $b=2$ and $c=0$, and it is NP-complete otherwise, even for simple graphs.
\end{theorem} 

\begin{proof}
In every case, the input graph $G$ has to be $(b+2c)$-regular, since otherwise it cannot cover $F(b,c)$. This condition can be checked in polynomial time. Next observe that a $(b+2c)$-regular graph $G$ covers $F(b,c)$ if and only if it contains $b$ pairwise disjoint perfect matchings whose removal leaves us with a $2c$-regular graph without semi-edges. Indeed, these matchings are the preimages of the $b$ semi-edges in a covering projection. The remaining $2c$-regular graph without semi-edges can be always partitioned into $c$ pairwise disjoint $2$-factors by the well known Petersen's theorem~\cite{k:LP86,petersen1891}, and each of the $2$-factors will cover one of the $c$ loops of $F(b,c)$. Note that a possible presence of loops in the input graph does not cause any problems.

The polynomially solvable cases then follow easily:
\begin{itemize}
  \item If $b=0$, the checking is trivial.
  \item If $b=1$, the existence of a perfect matching can be checked in polynomial time, for instance by a simple adaptation of Edmonds' blossom algorithm~\cite{edmonds1965paths}. The only thing we have to be careful about is that any possible semi-edges in the input graph have to be added to the perfect matching. This can be done by a simple preprocessing.
  \item If $b=2$ and $c=0$, $G$ itself has to be a $2$-regular graph and hence it contains two disjoint perfect matchings if and only if it contains at least one perfect matching, i.e., when all connected components of $G$ are even.
\end{itemize}
The NP-complete cases follow from Theorem~\ref{t-onevertex-general} by setting $k=b$ and $d=b+2c$.
\end{proof}

\section{Two-vertex target graphs}\label{s:twovertex}

Let $W(k,m,\ell,p,q)$ be the two-vertex graph with $k$ semi-edges and $m$ loops at one vertex, $p$ loops and $q$ semi-edges at the other one, and $\ell>0$ multiple edges connecting the two vertices (these edges are referred to as {\em bars}). In other words, $W(k,m,\ell,p,q)$ is obtained from the disjoint union of $F(k,m)$ and $F(q,p)$ by connecting their vertices by $\ell$ parallel edges.
For an example see the graph $H$ from Figure~\ref{fig:examplecover} which is isomorphic to both $W(1,1,2,1,0)$ and $W(0,1,2,1,1)$.

\begin{theorem}\label{t:2-vertex}
The {\sc $W(k,m,\ell,p,q)$-Cover} problem is solvable in polynomial time in the following cases
\begin{enumerate}
\item $k+2m\neq 2p+q$ and ($k\le 1$ or $k=2$ and $m=0$) and ($q\le 1$ or $q=2$ and $p=0$),
\item $k+2m = 2p+q$ and $\ell =1$ and $k=q\le 1$ and $m=p=0$,
\item $k+2m = 2p+q$ and $\ell >1$ and $k=m=p=q=0$,
\end{enumerate} 
and it is NP-complete otherwise.
\end{theorem}

Note that Case 1 applies to non-regular target graph $W$, while Cases 2 and 3 apply to regular graphs $W$, i.e., they cover all cases when $k+2m+\ell=2p+q+\ell$. 

We will refer to the vertex with $k$ semi-edges as {\em blue} and the vertex with $q$ semi-edges as {\em red}. In a covering projection $f=(f_V,f_E)$ from a graph $G$ onto $W(k,m,\ell,p,q)$, we view the restricted vertex mapping $f_V$ as a coloring of $V(G)$. We call a vertex $u\in V(G)$ blue (red) if $f_V$ maps $u$ onto the blue (red, respectively) vertex of $W(k,m,\ell,p,q)$. 
In order to keep the text clear and understandable, we divide the proof into a sequence of claims in separate subsections.
This will also allow us to state several hardness results in a stronger form.

\subsection{Polynomial parts of Theorem~\ref{t:2-vertex}}

We follow the case-distinction from the statement of Theorem~\ref{t:2-vertex}:
\begin{enumerate}
\item If $k+2m\neq 2p+q$, then the two vertex degrees of $W(k,m,\ell,p,q)$ are different,
and the vertex restricted mapping is uniquely defined for any possible graph covering projection from the input graph $G$ to $W(k,m,\ell,p,q)$. For this coloring of $G$, if it exists, we check if it is degree-obedient. 
If not, then $G$ does not cover $W(k,m,\ell,p,q)$. If yes, we check using Theorem~\ref{t-onevertex-general} whether the blue subgraph of $G$ covers $F(k,m)$ and whether the red subgraph of $G$ covers $F(q,p)$. 
If any one of them does not, then $G$ does not cover $W(k,m,\ell,p,q)$. 
If both of them do, then $G$ covers $W(k,m,\ell,p,q)$, since the
``remaining'' subgraph of $G$ formed by edges with one end-vertex red and the other one blue is $\ell$-regular and bipartite,
thus covering the $\ell$ parallel edges of $W(k,m,\ell,p,q)$ (Proposition~\ref{p:dob2}).
\item In Case~2, the input graph $G$ covers $W(1,0,1,0,1)$ only if $G$ is 2-regular.
If this holds, then $G$ is a disjoint union of cycles and paths with one semi-edge attached to each of its endpoints (such components are referred to as \emph{open paths}).
A cycle covers $W(1,0,1,0,1)$ if and only if its number of vertices is divisible by 4. An open path covers $W(1,0,1,0,1)$ if and only if its number of vertices is even.
For the subcase of $k=q=0$, see the next point.
\item The input graph $G$ covers $W(0,0,\ell,0,0)$ only if it is a bipartite $\ell$-regular graph without semi-edges, but in that case it does cover $W(0,0,\ell,0,0)$, as follows from Proposition~\ref{p:dob2}.  
\end{enumerate}

\subsection{NP-hardness for non-regular target graphs}

\begin{proposition}\label{t:2vertex-np-differentdegrees}
Let the parameters $k,m,p,q$ be such that $k+2m\neq 2p+q$, and
(($k\ge 3$ or $k=2$ and $m\ge 1$), or ($q\ge 3$ or $q=2$ and $p\ge 1$)). Then the {\sc $W(k,m,\ell,p,q)$-Cover} problem is NP-complete.
\end{proposition}

\begin{proof}
The parameters imply that at least one of the problems {\sc $F(k,m)$-Cover} and {\sc $F(q,p)$-Cover} is NP-complete by Section~\ref{s:onevertex}. Without loss of generality assume that this is the case of {\sc $F(q,p)$-Cover}.

Let $a := k+2m$ and $b := 2p+q$ and let $c$ be the smallest even number greater than both $a$ and $b$.  We shall construct a gadget which will be used in our reduction. We shall start with the construction for $\ell = 1$. 

We take two disjoint copies of $K_{c}$ and denote the vertices in the cliques as $x_1, \ldots, x_{c}$ and $y_1, \ldots, y_{c}$, respectively. Remove $(c-b-1)$ edge-disjoint perfect matchings, corresponding to $(c-b-1)$ color classes in some fixed $(c-1)$-edge-coloring of $K_{c}$, from the first copy of $K_{c}$, and remove $(c-a-1)$ edge-disjoint perfect matchings, corresponding to $(c-a-1)$ color classes in some fixed $(c-1)$-edge-coloring of $K_{c}$, from the second one. Add two new vertices $v,w$ and connect them by edges $vx_1$ and $wy_1$. Furthermore, add edges $x_iy_i$ for all $2 \leq i \leq c$. 

We denote the resulting graph by $G_{a,b}$. See Figure~\ref{fig:diffdegrees1} for an example. 

\begin{figure}[h]
\centering
\includegraphics[scale=0.9]{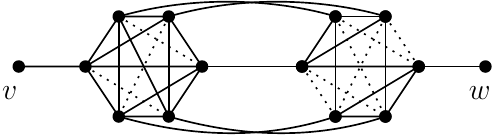}
\caption{A gadget $G_{3,4}$ from Proposition~\ref{t:2vertex-np-differentdegrees}.}
\label{fig:diffdegrees1}
\end{figure}

If $\ell > 1$, take $\ell$ disjoint copies of $G_{a,b}$ and denote their $v$-vertices as $v_1,\ldots, v_\ell$ and their $w$-vertices as $w_1,\ldots, w_\ell$. Furthermore, denote the corresponding vertices in the $j$-th copy ($1 \leq j \leq \ell$) of $G_{a,b}$ as $x_{j,1}, \ldots, x_{j,c}$ and $y_{j,1}, \ldots, y_{j,c}$. 

Insert edges between vertices $v_1,\ldots, v_{\ell}$ and $x_{1,1}, \ldots, x_{\ell,1}$ so that they induce a complete bipartite graph with one part being $v_1,\ldots, v_{\ell}$ and the other part being $x_{1,1}, \ldots, x_{\ell,1}$. The analogous construction will be done for $w_1,\ldots, w_{\ell}$ and $y_{1,1}, \ldots, y_{\ell,1}$. 
Moreover, for each $i \in \{2,\ldots,c \}$, insert edges between $x_{1,i}, \ldots, x_{\ell,i}$ and $y_{1,i}, \ldots, y_{\ell,i}$ so that they induce a complete bipartite graph with one part being $x_{1,i}, \ldots, x_{\ell,i}$ and the other part being $y_{1,i}, \ldots, y_{\ell,i}$. Denote the resulting graph as $G_{a,\ell,b}$ (for $\ell=1$, we set $G_{a,1,b}=G_{a,b}$). See Figure~\ref{fig:diffdegrees2} for an example.

\begin{figure}[h]
\centering
\includegraphics[scale=0.9]{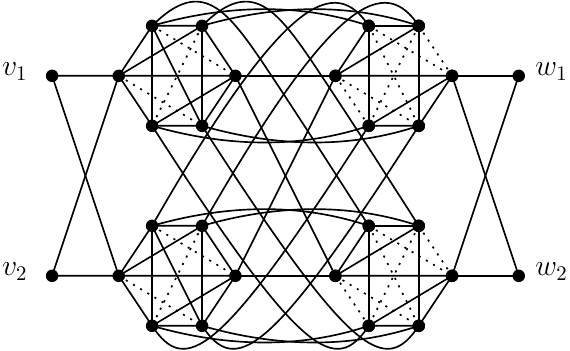}
\caption{A gadget $G_{3,2,4}$ from Theorem~\ref{t:2vertex-np-differentdegrees}.}
\label{fig:diffdegrees2}
\end{figure}

We will reduce from the problem {\sc $F(q,p)$-Cover}, which is NP-complete for these parameters by the results of the preceding section.
Let $G$ be an instance of {\sc $F(q,p)$-Cover} with $n$ vertices. Without loss of generality we may assume that $n$ is even. We shall construct a new graph $G'$ in the following way. 
Take $\ell$ copies of the graph $G$ and denote their vertices as $t_{j,1},\ldots, t_{j,n}$ in the $j$-th copy, respectively. Take $\ell$ copies of a graph with $n$ vertices that covers $F(k,m)$ (any $a$-regular bipartite graph on $n$ vertices will do) and denote their vertices as $u_{j,1},\ldots, u_{j,n}$ in the $j$-th copy, respectively.
  For each $h$, $1 \leq h \leq n$, take a new extra copy of $G_{a,\ell,b}$, denote their $v$ and $w$ vertices as $v_{h,1},\ldots, v_{h,\ell},w_{h,1},\ldots, w_{h,\ell}$ in the $h$-th copy, respectively, and identify $v_{h,j}$ with $u_{j,h}$ and $w_{h,j}$ with $t_{j,h}$ for each $1 \leq j \leq \ell$ and $1\le h \le n$.
Note that the constructed graph $G'$ is linear in the size of $G$. We claim that $G'$ covers $W(k,m,\ell,p,q)$ if and only if $G$ covers $F(q,p)$.

For the `only if' direction, suppose that $G'$ covers $W(k,m,\ell,p,q)$. First of all, because of the different degrees of the vertices of $W(k,m,\ell,p,q)$, we have a clear information about the vertex mapping part of the covering projection. In particular, the $v$ and $y$ vertices of the copies of $G_{a,\ell,b}$ are mapped onto the vertex of degree $a+\ell$ in $W(k,m,\ell,p,q)$, while the $x$ and $w$ ones are mapped onto the vertex of degree $b+\ell$. Hence the edges of each copy of $G$ must map onto the loops and half-edges incident with the vertex of degree $b+\ell$ in $W(k,m,\ell,p,q)$, and hence $G$ covers $F(q,p)$.  

Regarding the backward direction, the covering projection from $G'$ onto $W(k,m,\ell,p,q)$ is constructed as follows. Map the $v$ and $y$ vertices of the copies of $G_{a,\ell,b}$  onto the vertex of degree $a+\ell$ in $W(k,m,\ell,p,q)$, and the $x$ and $w$ ones  onto the vertex of degree $b+\ell$. This is a degree obedient vertex mapping of $V(G')$ onto the vertices of $W(k,m,\ell,p,q)$. The edges of $G'$ with one end-vertex of degree $a+\ell$ and the other one of degree $b+\ell$ induce a bipartite $\ell$-regular graph, and therefore can be mapped to the $\ell$ bars of $W(k,m,\ell,p,q)$ in a locally bijective way. If we delete these edges, $G'$ falls apart into several components of connectivity. The components induced by the $x$ vertices from copies of $G_{a,\ell,b}$ are $a$-regular $a$-edge colorable subgraphs of $G_{a,\ell,b}$ and hence their edges cover $F(k,m)$. The components induced by the $y$ vertices from copies of $G_{a,\ell,b}$ are $b$-regular $b$-edge colorable subgraphs of $G_{a,\ell,b}$ and hence their edges cover $F(q,p)$. The components induced by the $v$ vertices induce copies of the $a$-regular $a$-edge colorable graph chosen in the construction of $G'$, and hence they cover $F(k,m)$. Last but not least, the components induced by the $w$ vertices are isomorphic to $G$, whose edges cover $F(q,p)$ by the hypothesis of the `if' direction of the proof. Putting all these edge mappings together we obtain a covering projection from $G'$ onto $W(k,m,\ell,p,q)$, which concludes the proof.
\end{proof} 

\subsection{NP-hardness for connected regular target graphs}

The aim of this subsection is to conclude the proof of Theorem~\ref{t:2-vertex} by showing the NP-hardness for the case of $\ell\ge 1$ and $k+2m=2p+q$. We will actually prove a result which is more general in two directions. Firstly, we formulate the result in the language of colorings of vertices, and secondly, we prove the hardness for bipartite inputs.
This might seem surprising, as we have seen in Section~\ref{s:semi-edges} that bipartite graphs can make things easier. 
Moreover, this strengthening in fact allows us to prove the result in a unified, and hence simpler, way.

Note that the following definition of a relaxation of usual proper 2-coloring resembles the so-called \emph{defective 2-coloring} (see survey of Wood~\cite{arxiv1803.07694}). However, the definitions are not equivalent.

\begin{definition}
A {\em $(b,c)$-coloring} of a graph is a 2-coloring of its vertices such that every vertex has $b$ neighbors of its own color and $c$ neighbors of the other color.
\end{definition}

\begin{observation}\label{t:ab-coloring-bipartite-cover}
For any parameters $k,m,\ell,p,q$ such that $k+2m=2p+q$, a bipartite graph $G$ with no semi-edges covers $W(k,m,\ell,p,q)$ if and only if it allows a $(k+2m,\ell)$-coloring.
\end{observation}

\begin{proof}
On one hand, any graph covering projection from $G$ to $W(k,m,\ell,p,q)$ induces a $(k+2m,\ell)$-coloring of $G$, provided $k+2m=2p+q$. On the other hand, a $(k+2m,\ell)$-coloring of $G$ is a degree-obedient vertex mapping from $G$ to $W(k,m,\ell,p,q)$, again provided that $k+2m=2p+q$. If $G$ is bipartite and has no semi-edges, then this mapping can be extended to a graph covering projection by Theorem~\ref{p:dob3}.
\end{proof}

In view of the previous observation, we will be proving the NP-hardness results for the following problem:
\computationproblem
{$(a,b)$-Coloring}
{A graph $G$.}
{Does $G$ allow an $(a,b)$-Coloring?}

\begin{theorem}\label{t:ab-coloring}
For every pair of positive integers $b,c$ such that $b+c\ge 3$, the {\sc $(b,c)$-Coloring} problem is NP-complete even for simple bipartite graphs.
\end{theorem}

\begin{proof}
Since the proof of this result is relatively complex, we start by an outline of it.

First observe that the {\sc $(b,c)$-Coloring} and {\sc $(c,b)$-Coloring}
problems are polynomially equivalent on bipartite graphs,
as the colorings are mutually interchangeable
by switching the colors in one class of the bi-partition. Thus we may consider
only $b\ge c$.

\begin{figure}[b]
\begin{center}
\scalebox{0.9}{\includegraphics[page=2]{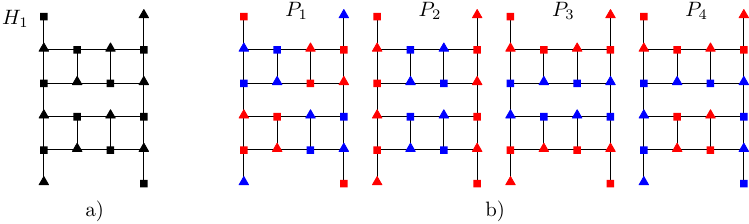}}
\end{center}
\caption{A 20-vertex auxiliary graph $H_1$, used in the first part of the proof of
Theorem~\ref{t:ab-coloring}, and its possible partial $(2,1)$-colorings.}\label{fig:malaab-2.0}
\end{figure}
\begin{figure}
\begin{center}
\scalebox{1.0}{\includegraphics[page=4]{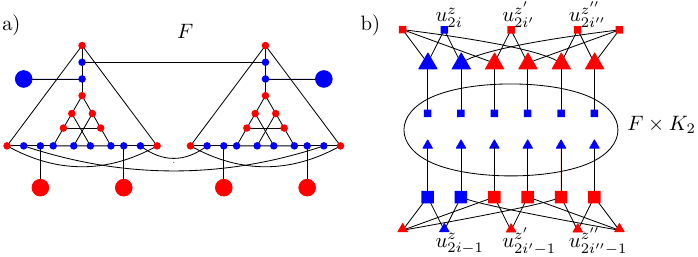}}
\caption{Garbage collection and the overall construction for the first part of Theorem~\ref{t:ab-coloring}.
Clause gadgets are in the corners of the figure b).} \label{fig:malaab-4.0}
\vspace{-3ex}
\end{center}
\end{figure}

NP-hardness of the {\sc $(2,1)$-Coloring} is proved in Lemma~\ref{t:ab-coloring-np-21} by a reduction from {\sc
NAE-3-SAT}~\cite{k:GJ79} by using three kinds of building blocks: a clause
gadget (here $K_{1,3}$),
a vertex gadget enforcing the same color on selected subset of vertices, and a
garbage collection that allows to complete the coloring to a cubic graph, that
as a part contains the vertex and clause gadgets linked together to represent a
given instance of {\sc NAE-3-SAT}. This reduction is the actual core of the
proof, and is briefly sketched in Figures~\ref{fig:malaab-2.0} and~\ref{fig:malaab-4.0}.
The former one shows a special gadget $H_1$ used in the color-enforcing
constructions of this reduction. For every variable, copies of $H_1$ are concatenated into a chain, whose one side is connected to ensure that the coloring $P_2$ (or its inverse) is the only admissible (2,1)-coloring, the other side transfers this information as the truth valuation of the variable to the clause gadgets of clauses containing it. The latter figure sketches the garbage collection and the overall
construction of the reduction.

The result on {\sc $(2,1)$-Coloring}, in particular, implies that {\sc $W(2,0,1,0,2)$-Cover} is NP-complete for simple bipartite input graphs, whereas the semi-edgeless dumbbell graph $W(0,1,1,1,0)$ is the smallest semi-edgeless graph whose covering is NP-complete.  

\begin{figure}[t]
\begin{center}
{\includegraphics[width=0.6\textwidth]{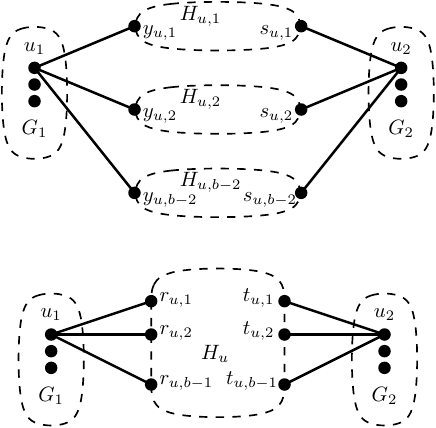}}
\end{center}\vspace{-1ex}
\caption{An illustration of the constructions used in the proof of
Theorem~\ref{t:ab-coloring}; a reduction to {\sc $(b,1)$-Coloring} at the top, and a reduction to {\sc $(b,c)$-Coloring} at the bottom.}
\label{fig:joinsk}
\end{figure}

Further on, we reduce {\sc $(2,1)$-Coloring} to {\sc $(b,1)$-Coloring}  in the proof of Proposition~\ref{t:ab-coloring-np-malaab} by using two copies of the instance of {\sc $(2,1)$-Coloring}
and linking them together by suitable graphs called bridges, that enforce replication of colors for the desired coloring. 
See a brief sketch in Figure~\ref{fig:joinsk} (left).
In view of the initial observation, at this point we know that 
{\sc $(b,1)$-Coloring} and {\sc $(1,b)$-Coloring} are NP-complete on bipartite inputs for all $b\ge 2$.

Then we reduce {\sc $(1,c)$-Coloring} to {\sc $(b,c)$-Coloring} with $b> c$ in Propositions~\ref{p:k>l+1coloring} and~\ref{p:k=l+1coloring}. 
Again we take two copies of an instance of {\sc $(1,c)$-Coloring}, say a $(1+c)$-regular graph $G$. 
As sketched in Figure~\ref{fig:joinsk} (right), we construct an auxiliary graph $H$ with two vertices of degree $b-1$ (called the  ``connector'' vertices), all other vertices being of degree $b+c$ (these are called the ``inner'' vertices). This bridge graph is such that in every two-coloring of its vertices, such that all inner vertices have exactly $b$ neighbors of their own color and exactly $c$ neighbors of the opposite color, while the connector vertices have at most $b$ neighbors of their own color and at most $c$ neighbors of the opposite color, in every such coloring the connector vertices and their neighbors always get the same color. And, moreover, such coloring exists. We then take two copies of $G$ and for every vertex of $G$, identify its copies with the connector vertices of a copy of the bridge graph (thus we have as many copies of the bridge graph as is the number of vertices of $G$). The above stated properties of the bridge graph guarantee that the new graph allows a $(b,c)$-coloring if and only if $G$ allows a $(1,c)$-coloring.

It is worth mentioning that we have provided two different constructions of the bridge gadget. A general one for the case of $b\ge c+2$ and a specific one for the case of $b=c+1$. It is a bit surprising that the case analysis needed to prove the properties of the bridge graph is much more involved for the specific construction in the case of $b=c+1$.  

Finally, for {\sc $(b,b)$-Coloring} with $b\ge 2$ we establish a completely different
reduction in the proof of Proposition~\ref{t:aa-coloring-np}. We reduce from a special variant of satisfiability ($k$-in-$2k$)-SAT$_q$, a generalization of {\sc NAE-3-SAT}. 
\end{proof}

Theorem~\ref{t:ab-coloring} and Observation~\ref{t:ab-coloring-bipartite-cover} imply the following proposition, which concludes the proof of Theorem~\ref{t:2-vertex}.

\begin{proposition}
The {\sc $W(k,m,\ell,p,q)$-Cover} problem is NP-complete for simple bipartite input graphs for all parameter sets such that $k+2m=2p+q\ge 1$, $\ell\ge 1$, and $k+2m+\ell\ge 3$. 
\end{proposition}

The rest of this subsection is devoted to a detailed proof of Theorem~\ref{t:ab-coloring}.

\begin{observation}\label{o:ab-ba-coloring}
A bipartite graph $G$ allows a $(b,c)$-coloring if and only if it allows a $(c,b)$-coloring.
\end{observation}

\begin{proof}
Let $A$ and $B$ be the classes of bi-partition of $V(G)$ and assume that $G$ has a $(b,c)$-coloring using red and blue colors. By swapping these colors on the set $B$ we obtain a $(c,b)$-coloring.
\end{proof}

\begin{corollary}\label{c:ab-ba-coloring}
The problems {\sc $(b,c)$-Coloring} and {\sc $(c,b)$-Coloring} are polynomially equivalent on bipartite graphs. 
\end{corollary}

\begin{proposition}\label{t:ab-coloring-np-malaab}
For every $b\ge 2$, the problem {\sc $(b,1)$-Coloring} is NP-complete even for simple bipartite graphs. 
\end{proposition}

We will develop the proof as a series of claims.
We first consider $(2,1)$-coloring of cubic bipartite graphs. Through our arguments the classes of bi-partition will be indicated in figures by vertex shapes --- squares and triangles, while for the $(2,1)$-coloring we use red and blue colors.

Observe first 
that whenever a graph $G$ contains a $C_4$ as an induced subgraph then in any $(2,1)$-coloring of $G$ 
it is impossible to color exactly three vertices of the $C_4$ by the same color. 
The reason is that in such a case the remaining vertex would be adjacent to two vertices of the opposite color, which is not allowed. By the same argument we deduce that if both colors are used on the $C_4$ then vertices of the same color are adjacent.

The following two observations are immediate.

\begin{figure}
\begin{center}
\scalebox{1.1}{\includegraphics[page=1]{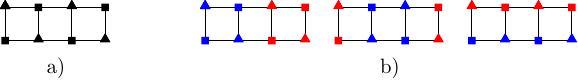}}
\end{center}
\caption{Partial $(2,1)$-colorings of an 8-vertex auxiliary subgraph.}\label{fig:malaab-1}
\end{figure}

\begin{observation}
Whenever a graph $G$ contains as a subgraph the graph on 8 vertices depicted in Figure~\ref{fig:malaab-1} a) then in any in any $(2,1)$-coloring of $G$ the color classes match on the subgraph one of the three patterns depicted in red and blue colors in Figure~\ref{fig:malaab-1} b).
\end{observation}

\begin{figure}
\begin{center}
\scalebox{0.9}{\includegraphics[page=2]{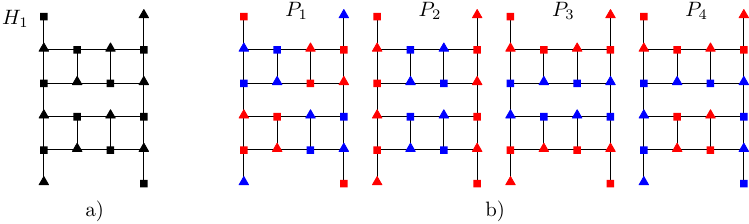}}
\end{center}
\caption{A 20-vertex auxiliary graph $H_1$ and its possible partial $(2,1)$-colorings.}\label{fig:malaab-2}
\end{figure}

\begin{observation}
Whenever a graph $G$ contains as a subgraph the graph $H_1$ on 20 vertices depicted in Figure~\ref{fig:malaab-2} a) then in any 
in any $(2,1)$-coloring of $G$ the color classes match on $H_1$ one of the four patterns depicted in red and blue colors in Figure~\ref{fig:malaab-2} b).
\end{observation}

\begin{figure}
\begin{center}
\scalebox{0.9}{\includegraphics[page=3]{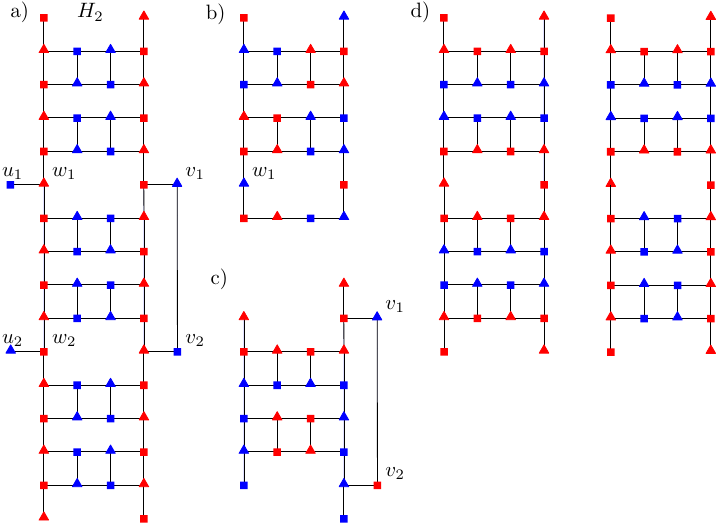}}
\end{center}
\caption{Forcing the same color on $u_1$ and $u_2$.}\label{fig:malaab-3}
\end{figure}

\begin{lemma}
If a cubic graph $G$ contains the graph $H_2$ (depicted in Figure~\ref{fig:malaab-3}~a)) as a subgraph, then in any $(2,1)$-coloring of $G$ the vertices $u_1$ and $u_2$ have the same color and their neighbors $w_1$, $w_2$ have the opposite color.
\end{lemma}

\begin{proof}
The graph $H_2$ contains three induced copies of $H_1$. If the pattern $P_1$ of Figure~\ref{fig:malaab-2} b) was used on some copy, then the same pattern must be used on all three copies.
Consequently, the vertex $w_1$ has two neighbors of the opposite color as indicated in Figure~\ref{fig:malaab-3} b), which is not allowed. This excludes the pattern $P_1$ from our reasoning.

If the pattern $P_4$ was used on the middle copy of $H_1$, then the vertices $v_1$ and $v_2$ have two neighbors of the opposite color as indicated in Figure~\ref{fig:malaab-3} c), which is also not allowed.

Therefore the middle copy of $H_1$ uses either pattern $P_2$ or $P_3$ and the claim follows. Note that both patterns might be used on the same $H_2$ see, Figure~\ref{fig:malaab-3}~a) and d).
\end{proof}

\begin{lemma}\label{t:ab-coloring-np-21}
The problem {\sc $(2,1)$-Coloring} is NP-complete even for simple cubic bipartite graphs. 
\end{lemma}

\begin{proof}
\sloppy
We reduce from a well known NP-complete variant of the  {\sc Monotone-NAE-3-SAT} problem, which given a formula $\phi$ in CNF without negations, consisting of clauses $C_1, \ldots, C_m$, where each $C_j$ is a disjunction of exactly three distinct literals, asks whether $\phi$ has a truth assignment such that each clause contains at least one positively valued literal, as well as at least one negatively valued one. In fact, this is equivalent to bicolorability of 3-uniform hypergraphs. The NP-completeness of this problem is proven in~\cite{Schaefer78}.
\fussy

For a given $\phi$ we build a bipartite cubic graph $G$ that allows a $(2,1)$-coloring if and only if $\phi$ allows the required assignment. The graph has several functional blocks: variable gadgets, clause gadgets enforcing the valid truth assignment already for a partial $(2,1)$-coloring and also garbage collection allowing to extend the partial coloring to the entire cubic graph. By partial $(2,1)$-coloring we mean a restriction of a $(2,1)$-coloring to a subgraph, i.e. a vertex 2-coloring where every vertex has at most two neighbors of its own color and at most one neighbor of the other color. 

For a variable $z$ that in $\phi$ has $k$ occurrences, we build a variable gadget consisting of a cyclic chain of $2k$ graphs $H_1$ linked together with further vertices $u_i^z$ and $v_i^z$ so each three consecutive copies of $H_1$ induce the graph $H_2$ of Fig~\ref{fig:malaab-3} a). In this gadget the color of $u_1^z,\dots,u_{2k}^z$ represent the truth assignment of $z$.

The clause gadget for a clause $C_j$ is a claw $K_{1,3}$. When a variable $z$ occurs in a clause $C_j$ we add an edge between an $u_{2i}^z$ and a unique leaf of the clause gadget $K_{1,3}$ so that each clause gadget is linked to a distinct $u_{2i}^z$. 

Observe that any partial $(2,1)$-coloring of the so far formed graph corresponds to the valid truth assignments and vice-versa: leaves of each clause gadget $K_{1,3}$ are not monochromatic, while the edges added between the vertex and clause gadget have both end of the same color as each $u_{2i}^z$ has already a neighbor $v_{2i}^z$ of the other color.

It remains to extend the graph to a cubic graph so that the partial $(2,1)$-coloring is preserved within a ``full'' $(2,1)$-coloring. We first add additional copies of clause gadgets and link them to the vertex gadgets by the same process, so that each $u_{2i}^z$ is linked to exactly two clause gadgets. We then repeat the same process twice for vertices $u_{2i-1}^z$ with odd valued indices. Now the only vertices that do not have degree three are the former leaves of clause gadgets, where each of them is now of degree two.

\begin{figure}
\begin{center}
\scalebox{0.95}{\includegraphics[page=4]{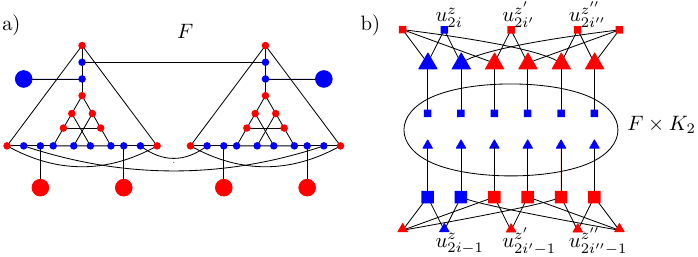}}
\caption{Garbage collection and the overall construction for Theorem~\ref{t:ab-coloring-np-21}.
Clause gadgets are in the corners of the part b).} \label{fig:malaab-4}
\end{center}
\end{figure}

For this purpose we involve an auxiliary graph $F$ and one of its partial $(2,1)$-colorings depicted in Figure~\ref{fig:malaab-4} a). For each clause $C_j$ we take a copy of the bipartite graph $F\times K_2$ and
merge its 12 vertices of degree one with the twelve vertices of degree two stemming from the four copies of the clause gadgets as shown in Figure~\ref{fig:malaab-4} a). The merged vertices are indicated by big symbols.

This step completes the construction of the desired simple cubic bipartite graph $G$ that allows a  $(2,1)$-coloring if and only if $\phi$ allows a not-all-equal truth assignment. The way how such truth assignment can be derived from a $(2,1)$-coloring has been already discussed. In the opposite way, the truth assignment yields a coloring of the vertex gadgets, say the color blue would represent variables assigned true, while red would represent variable assigned false. Then the coloring can be completed to clause gadgets and auxiliary graphs $F\times K_2$ by using patterns depicted in Figure~\ref{fig:malaab-4}. In the last step we involve the standard lift of a coloring to a product, namely that  the same color is used on the two copies of a vertex in the $F\times K_2$ as the original vertex has in $F$.
\end{proof}

\begin{proof}[Proof of Proposition~\ref{t:ab-coloring-np-malaab}]
For $b\ge 3$ we reduce the {\sc $(2,1)$-Coloring} to {\sc $(b,1)$-Coloring}. Let $G$ be a bipartite cubic graph whose $(2,1)$-coloring has to be decided. 

First we construct an auxiliary graph $F$ consisting of two disjoint unions of $K_{b,b}$ with classes of bi-partition $A_1,B_1,A_2$ and $B_2$ that are joined together by two perfect matchings, one between sets $A_1$ and $A_2$ and the other between $B_1$ and $B_2$. Finally, we add two vertices 
$u$ and $v$, make $u$ adjacent to some $u'\in A_1$ and $v$ adjacent to some $v'\in B_1$ and 
remove the edge $(u',v')$.

We claim that in any partial $(b,1)$-coloring of $F$ the vertices $u,v,u'$ and $v'$ receive the same color.
Observe first that the complete bipartite graph $K_{b,b}$ on $A_2$ and $B_2$ is monochromatic as otherwise one vertex would have at least two neighbors of the opposite color. Now each vertex of $A_2$ and $B_2$ has $a$ neighbors of the same color, say red, so the sets $A_1$ and $B_1$ are blue. The vertex $u'$ now has a single red neighbor and $b-1$ blue neighbors so $u$ is blue as well. Analogously for $v$ and $v'$.

We take two copies $G_1$ and $G_2$ of the graph $G$ and for each $w\in V_G$ we insert $b-2$ copies $F^w_1,\dots,F^w_{b-2}$ of the graph $F$, where we identify $w_1$ with $u^w_1,\dots,u^w_{b-2}$ and also $w_2$ with $v^w_1,\dots,v^w_{b-2}$. By this process we get a bipartite $(b+1)$-regular graph $H$.

The fact that graph $H$ allows an $(b,1)$-coloring if and only if $G$ allows an $(2,1)$-coloring follows from the 
fact that the all $b-2$ neighbors of any $w_1$ outside $G_1$, i.e. inside the copies of $F$, have the same color as $w_1$.
\end{proof}

\begin{proposition}\label{p:k>l+1coloring}
For every $c\ge 2$ and $b\ge c+2$, the {\sc $(b,c)$-Coloring} problem is NP-complete even for simple bipartite graphs.
\end{proposition}

\begin{figure}[h]
\begin{center}
\scalebox{0.9}{\includegraphics{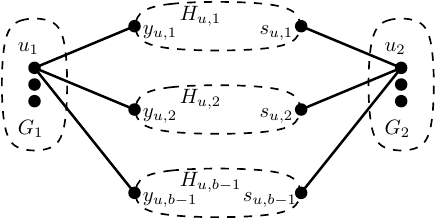}}
\end{center}
\caption{An illustration to the construction of graph $G'$ from Proposition~\ref{p:k>l+1coloring}.}
\label{fig:k>l+1H}
\end{figure}

\begin{proof}
We will prove that {\sc $(1,c)$-Coloring} polynomially reduces to {\sc $(b,c)$-Coloring} for simple bipartite inputs. Given a simple bipartite $(1+c)$-regular graph $G$ as input of  {\sc $(1,c)$-Coloring}, construct a graph $G'$ by taking two disjoint copies $G_1,G_2$ of $G$ and connecting them by ``bridges'' as follows. Let $H$ be a graph with two pendant vertices $x,t$ of degree 1 and all other vertices of degree $b+c$. Let $y$ be the neighbor of $x$ and $s$ the neighbor of $t$ in $H$. The vertices of degree $b+c$ in $H$ will be called its {\em inner vertices}. Let the companion vertices of $G_1$ and $G_2$ that are copies of a vertex $u$ of $G$ be denoted by $u_1$ and $u_2$, respectively. For every vertex $u\in V(G)$, take $b-1$ copies $H_{u,i},i=1,2,\ldots,b-1$ of $H$, with vertices of $H_{u,i}$ denoted by $z_{u,i}$, for $z\in V(H)$. For every $u\in V(G)$, identify the vertices $x_{u,i}, i=1,2,\ldots, b-1$ with the vertex $u_1$ and  identify the vertices $t_{u,i}, i=1,2,\ldots, b-1$ with the vertex $u_2$. See an illustration in Figure~\ref{fig:k>l+1H}.

\begin{lemma}\label{l:Hink>l+1abcolor}
Suppose that the number of inner vertices of $H$ is divisible by 4. Let $\varphi:V(H)\longrightarrow\{red,blue\}$ be a red-blue coloring of the vertices of $H$ such that every inner vertex has exactly $b$ neighbors of its own color (and hence $c$ neighbors of the other color). Then either $\varphi(x)=\varphi(y)=\varphi(s)=\varphi(t)$, or $\varphi(x)=\varphi(s)\neq\varphi(y)=\varphi(t)$.
\end{lemma}

\begin{figure}[h]
\begin{center}
\scalebox{0.9}{\includegraphics{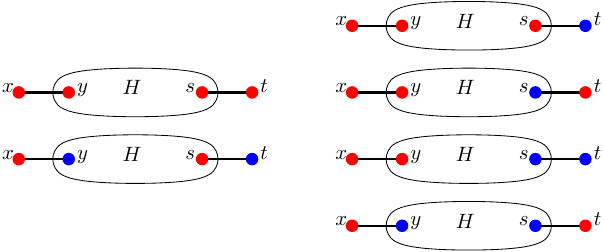}}
\end{center}
\caption{Feasible (in the left) and infeasible (in the right) red-blue colorings of a bridge graph $H$.}\label{fig:k>l+1Hfeasiblecolorings}
\end{figure}

\begin{proof}
Let $\alpha$ be the number of inner vertices that are colored red, and let $\beta$ be the number of inner vertices that are colored blue. Every red inner vertex has $c$ blue neighbors, and so $H$ has $\alpha c$ red-blue edges, with at most two of them being the pendant ones. Similarly, $H$ has  $\beta c$ red-blue edges, with at most two of them being the pendant ones. Hence 
$$\alpha c-\epsilon_r=\beta c-\epsilon_b$$
for some $\epsilon_r,\epsilon_b\in\{0,1,2\}$
(with some additional restrictions, e.g., $\epsilon_r,\epsilon_b$ cannot be both equal to 2, but that is not important). Therefore,
$$|(\alpha-\beta)c|\le 2.$$
If $c>2$, this immediately implies $\alpha=\beta$. If $c=2$, we might get $|\alpha-\beta|=1$, but then $\alpha$ and $\beta$ would be of different parities, contradicting the assumption of $\alpha+\beta$ being even. We conclude that $\alpha=\beta$, and that $\alpha$ is even.

Suppose $x$ and $y$ have the same color, say the red one. Then both $s$ and $t$ must be red as well, because
\begin{itemize}
\item $\varphi(s)=red,\varphi(t)=blue$ would yield $\alpha c-1=\beta c$, which is impossible,
\item $\varphi(s)=blue,\varphi(t)=red$ would yield $\alpha c=\beta c-1$, which is impossible, and
\item $\varphi(s)=\varphi(t)=blue$ would imply that the red subgraph of $H$ has an odd number of vertices of odd degree (either 1, if $b$ is even, or $\alpha+1$ if $b$ is odd), which is impossible by the well known Handshaking lemma.  
\end{itemize}  

Let $x$ and $y$ have different colors, say $x$ is red and $y$ is blue. Then $s$ and $t$ cannot have the same color by an argument symmetric to the one above. We cannot have $s$ blue and $t$ red, since $\alpha c=\beta c-2$ in such a case, which is not possible since $\alpha+\beta$ is divisible by 4. Hence $s$ must be red and $y$ blue. This concludes the proof of Lemma~\ref{l:Hink>l+1abcolor}. (See Figure~\ref{fig:k>l+1Hfeasiblecolorings} for an illustration of feasible and infeasible colorings of $H$.)
\end{proof}

\begin{lemma}\label{l:existujeHink>l+1}
For every $c\ge 2$ and $b\ge c+2$, there exists a bridge graph $H$ whose number of inner vertices is divisible by 4 and which allows a $(b,c)$-coloring such that all four vertices $x,y,s,t$ have the same color. 
\end{lemma}

\begin{proof}
Take two disjoint copies of the complete bipartite graph $K_{b,b}$. Let the classes of bi-partition in one of them be $A$ and $B$, and the classes of bi-partition in the other one $C$ and $D$. Pick an edge $ys$ such that $y\in A$ and $s\in B$. Add $c$ disjoint perfect matchings between $A$ and $C$, and $c$ other disjoint perfect matching between $B$ and $D$. In this way we have obtained a $(b+c)$-regular bipartite graph with $4b$ vertices. Delete the edge $ys$ and add pendant edges $xy$ and $st$ with new extra vertices $x,t$ of degree 1 to obtain the desired bridge graph $H$. Indeed, coloring $A\cup B\cup\{x,t\}$ red and $C\cup D$ blue yields a $(b,c)$-coloring such that the 4 vertices $x,y,s,t$ get the same color. See an illustrative example in Figure~\ref{fig:existsHink>l+1}. 
\end{proof}

\begin{figure}[h]
\begin{center}
\scalebox{0.9}{\includegraphics{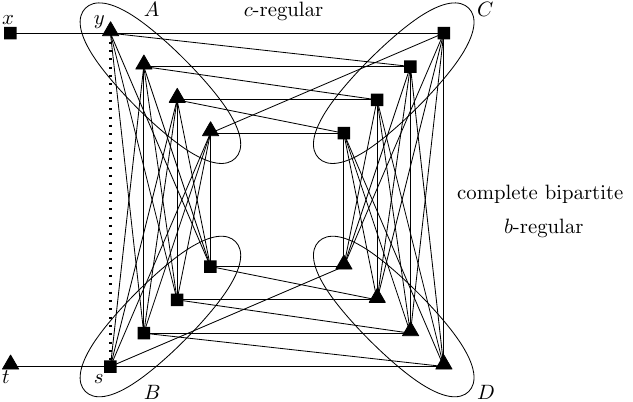}}
\end{center}
\caption{An example of the bridge graph $H$ for $b=4$ and $c=2$.}\label{fig:existsHink>l+1}
\end{figure}

Let us return to the proof of Proposition~\ref{p:k>l+1coloring}. Given a simple bipartite graph $G$, we construct $G'$ as described using the bridge graph $H$ from Lemma~\ref{l:existujeHink>l+1}. This $G'$ is simple, and since $H$ was created from a bipartite graph, $G'$ is bipartite as well. The proof of the proposition now follows from the following lemma.

\begin{lemma}
The graph $G'$ allows a $(b,c)$-coloring if and only if $G$ allows a $(1,c)$-coloring.
\end{lemma}

\begin{proof}
Suppose $G'$ allows a $(b,c)$-coloring, say $\varphi$. Consider a vertex $u\in V(G)$. Lemma~\ref{l:Hink>l+1abcolor} implies that either 
\begin{itemize}
\item $\varphi(u_1)=\varphi(y_{u,i})=\varphi(s_{u,i})=\varphi(u_2)$ for all $i=1,2,\ldots,b-1$, or
\item $\varphi(u_1)=\varphi(s_{u,i})\neq\varphi(y_{u,i})=\varphi(u_2)$ for all $i=1,2,\ldots,b-1$.
\end{itemize}
But the latter would mean that $u_1$ has $b-1>c$ neighbors of the opposite color, which is too many. Hence every vertex $u_1$ has $b-1$ neighbors of its own color in the bridge graphs, and therefore the restriction of $\varphi$ to $G_1$ is a $(1,c)$-coloring of $G_1$ (which is isomorphic to $G$).

On the other hand, if $G$ allows a $(1,c)$-coloring, use the same coloring on $G_1$ and $G_2$ and color the bridges so that for every $u\in V(G)$, both $u_1$ and $u_2$ have all their $b-1$ neighbors in the bridge graphs colored with their own color. This is possible by Lemma~\ref{l:existujeHink>l+1}, and this gives a $(b,c)$-coloring of $G'$. 
\end{proof}

The lemma now implies Proposition~\ref{p:k>l+1coloring}.
\end{proof}

\begin{proposition}\label{p:k=l+1coloring}
For every $c\ge 2$, the {\sc $(c+1,c)$-Coloring}  problem is NP-complete even for simple bipartite graphs.
\end{proposition}

\begin{proof}
We will prove that {\sc $(1,c)$-Coloring} polynomially reduces to {\sc $(c+1,c)$-Coloring} for simple bipartite inputs. Given a simple bipartite $(1+c)$-regular graph $G$ as input of  {\sc $(1,c)$-Coloring}, construct a graph $G'$ by taking two disjoint copies $G_1,G_2$ of $G$ and connecting them by ``bridges'', similarly as in the proof of Proposition~\ref{p:k>l+1coloring}. But this time we will describe the bridge graph $H$ explicitly from the very beginning of the proof. It has $4(c+1)$ ``inner'' vertices of degree $2c+1$ and two ``connector'' vertices of degree $c$. The inner part of $H$ is created from two copies of the complete bipartite graph $K_{c+1,c+1}$ whose classes of bi-partition are connected by cocktail-party graphs (i.e., complete bipartite graphs minus a perfect matching), and in one of the copies $c$ independent edges are deleted and replaced by edges leading to the connector vertices. The graph is illustrated in Figure~\ref{fig:existsHink=l+1}, but since we will heavily rely on its structure in the proof of its properties, we also describe it formally:

\begin{figure}[t]
\begin{center}
\scalebox{0.9}{\includegraphics{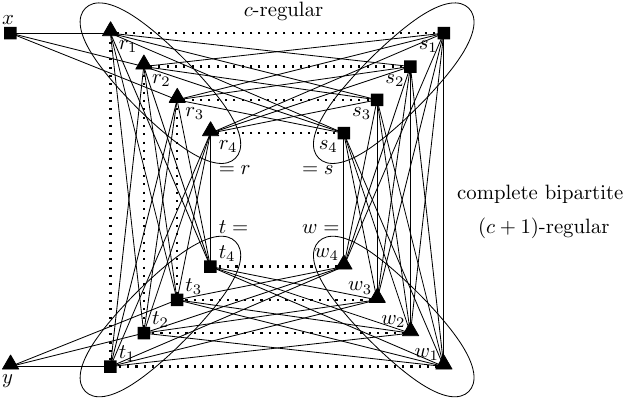}}
\end{center}
\caption{An example of the bridge graph $H$ for $c=3$.}\label{fig:existsHink=l+1}
\end{figure}

\begin{equation}
\begin{aligned}
V(H) = & \{x,y\}\cup\bigcup_{i=1}^{c+1}\{r_i,s_i,t_i,w_i\}, \\
E(H)=  & \bigcup_{i=1}^{c}\{xr_i,yt_i\} \cup \Bigg(\bigcup_{i,j=1}^{c+1}\{r_it_j\}\setminus\bigcup_{i=1}^{c}\{r_ir_i\}\Bigg) \\
       & \cup \bigcup_{i,j=1}^{c+1}\{s_iw_j\} \cup \Bigg(\bigcup_{i,j=1}^{c+1}\{r_is_j,t_iw_j\}\setminus\bigcup_{i=1}^{c+1}\{r_is_i,t_iw_i\}\Bigg)
\end{aligned}
\end{equation}
where for the sake of brevity, but also to stress their special roles, we write $r=r_{c+1},  s=s_{c+1}, t=t_{c+1}$ and $w=w_{c+1}$.

In the construction of $G'$, for every $u\in V(G)$, let the companion vertices in $G_1$ and $G_2$ which are copies of $u$ be again denoted by $u_1$ and $u_2$, respectively. We take a copy $H_u$ of $H$ and unify its connector vertices with $u_1$ and $u_2$. See an illustrative example in Figure~\ref{fig:k=l+1}. Note finally, that $G'$ is a bipartite graph, since $H$ is bipartite and the distance between $x$ and $y$ in $H$ is odd.

\begin{figure}
\begin{center}
\scalebox{0.9}{\includegraphics{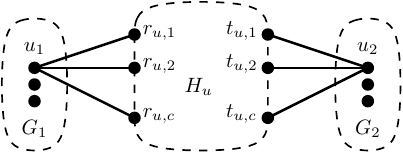}}
\end{center}
\caption{An illustration to the construction of $G'$.}\label{fig:k=l+1}
\end{figure}

\begin{lemma}\label{l:Hink=l+1abcolor}
Let $\varphi:V(H)\longrightarrow\{red,blue\}$ be a red-blue coloring of the vertices of $H$ such that every inner vertex has exactly $c+1$ neighbors of its own color (and hence $c$ neighbors of the other color). Then $\varphi(x)=\varphi(r_i)=\varphi(t_i)=\varphi(y)$ for all $i=1,2,\ldots,c$.
\end{lemma}

\begin{proof}
Suppose $\varphi(x)=red$. We will prove the result by a  case analysis. In the illustrative Figure~\ref{fig:caseanalysisHink=l+1}, the assumptions of the cases are marked with dark red and blue, the colorings that are derived from them by light red and blue, and the vertices that cause contradictions are stressed by arrows.

\begin{figure}[h]
\begin{center}
\scalebox{0.85}{\includegraphics[page=3]{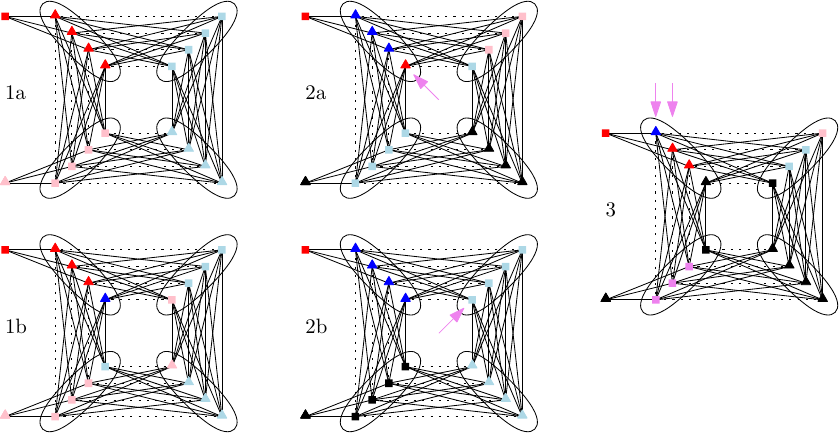}}
\end{center}
\caption{An illustration to the case analysis of Lemma~\ref{l:Hink=l+1abcolor}.}\label{fig:caseanalysisHink=l+1}
\end{figure}

\medskip\noindent
\emph{\textbf{Case 1} $\varphi(r_i)=red$ for all $i=1,2,\ldots,c$\\
\textbf{Subcase 1a} $\varphi(r)=red$}\\
In this case any two vertices $s_i,s_j$ will end up with the same number of red neighbors, regardless of the colors on $w_1,\ldots,w_{c+1}$. Therefore all $s_i$'s must have the same color. Every vertex $w_i$ then already has $c+1$ neighbors of this color among $s_i$'s, and thus all vertices $w_i$ have the same color as the $s_i$'s. If this color were red, every $s_i$ would have $2c+1$ red neighbors and no blue ones. Hence $\varphi(s_i)=\varphi(w_i)=blue$ for all $i=1,2,\ldots,c+1$. Then each $w_i$ has already $c+1$ neighbors of its own color, and so all the other neighbors (i.e., the vertices $t_i$, $i=1,2,\ldots,c+1$) are red. Now $t_1$ has  only $c$ red neighbors among the $r_i$'s, and therefore $y$ must be red as well.\\
\emph{\textbf{Subcase 1b} $\varphi(a)=blue$}\\
In this case, each $s_i$, $i=1,2,\ldots,c$ will end up seeing less red neighbors than $s$, regardless of the colors on $w_i$'s ($s$ has a red neighbor $r_i$, while $r_i$ is not a neighbor of $s_i$, and the private neighbor $r$ of $s_i$ is blue). Hence $s$ must be red and all $s_i$, $i=1,2,\ldots,c$ are blue. To supply the $s_i$'s with correct numbers of red neighbors, exactly one of the $w_i$'s must be red, all others are blue. The red one has just one red neighbor among $s_i$'s, and hence at least $c$ of the $t_i$'s are red. The blue vertices among $w_i$'s have $c$ blue neighbors among $s_i$'s, and so at least one of the $t_i$'s is blue. It follows that $\varphi(w_i)\neq\varphi(t_i)$ for all $i=1,2,\ldots,c+1$. Since every $r_i$, $i=1,2,\ldots,c$ has two red neighbors $x$ and $s$, it should have only (and exactly) $c-2$ red neighbors among $t_i$'s, and hence $\varphi(t_i)=\varphi(r_i)=red$ for $i=1,2,\ldots,c$. Then $\varphi(t)=blue$. Since $t_1$ has so far $c$ red neighbors ($c-1$ among $r_i$'s and one among $w_i$'s), $y$ must be red.

\medskip\noindent
\emph{\textbf{Case 2} $\varphi(r_i)=blue$ for all $i=1,2,\ldots,c$}\\
\emph{\textbf{Subcase 2a} $\varphi(r)=red$}\\
Any two $s_i,s_j$, $i,j=1,2,\ldots,c$ will end up with the same number of red neighbors (regardless the coloring of the $w_i$'s), and hence all $s_i$, $i=1,2,\ldots,c$ have the same color. Since $r$ is not a neighbor of $s$, $s$ will end up with less red neighbors than $s_1$. Therefore, $\varphi(s_i)=red$ for $i=1,2,\ldots,c$, and $\varphi(s)=blue$. Since $x$ is red, every $r_i$, $i=1,2,\ldots,c$ must have $c$ blue neighbors among the $t_i$'s, and because $c\ge 2$, it follows that all $t_i$'s (including $t=t_{c+1}$) are blue. 
But then the red vertex $r$ has too many ($c+1$) blue neighbors, a contradiction.
\\
\emph{\textbf{Subcase 2b} $\varphi(a)=blue$}\\
Any two $s_i$ vertices will end up with the same number of red neighbors, and hence all $s_i$'s (including $s$) have the same color, and this color must be blue, since a blue vertex $r_1$ would have $c+1$ red neighbors otherwise. Now every $w_i$ has already $c+1$ blue neighbors (the $s_i$'s), and thus all $w_i$'s are blue. But this causes a contradiction, since now each $s_i$ has all $2c+1$  neighbors blue.

\medskip\noindent
\emph{\textbf{Case 3} At least one of the $r_i$'s for $i=1,2,\ldots,c$ is red and at least one of them is blue}\\
Consider $i$ and $j$ such that $\varphi(r_i)=\varphi(r_j)$. Regardless the coloring of $w_i$'s, the vertices $s_i$ and $s_j$ will end up with the same number of red neighbors, and hence $\varphi(s_i)=\varphi(s_j)$. If, on the other hand, $\varphi(r_i)\neq\varphi(r_j)$, say $\varphi(r_i)=red$ and $\varphi(r_j)=blue$, then $s_i$ will end up with less red neighbors than $s_j$, and hence $\varphi(s_i)=blue$ and $\varphi(s_j)=red$. We conclude that for every $i=1,2,\ldots,c+1$, $r_i$ and $s_i$ get different colors.

Now consider two vertices $t_i,t_j$, $i,j=1,2,\ldots,c$. If $\varphi(r_i)=\varphi(r_j)$, then $r_i$ and $r_j$ have the same number of red neighbors among $\{x\}\cup\{s_1,s_2,\ldots,s_{c+1}\}\cup(\{t_1,t_2,\ldots,t_{c+1}\}\setminus\{t_i,t_j\})$. In order to end up with the same number of red neighbors in total, it must be $\varphi(t_i)=\varphi(t_j)$. If $r_i$ and $r_j$ got different colors, say $\varphi(r_i)=red$ and $\varphi(r_j)=blue$, then among $\{x\}\cup\{s_1,s_2,\ldots,s_{c+1}\}\cup(\{t_1,t_2,\ldots,t_{c+1}\}\setminus\{t_i,t_j\})$, $r_i$ has one more red neighbors than $r_j$. But the same difference should apply to the total number of red neighbors of $r_i$ and $r_j$, and hence $\varphi(t_i)=\varphi(t_j)$. We conclude that all vertices $t_j,j=1,2,\ldots,c$ have the same color. Since the graph $H$ is symmetric, this is either Case 1 or Case 2 from the standpoint of the $t_i$'s. These cases have already been treated and either they lead to a contradiction, or they require that all vertices $r_i$, $i=1,2,\ldots,\ell$ get the same color. Which contradicts the assumption of Case 3.   
\end{proof}

To conclude the proof of Proposition~\ref{p:k=l+1coloring}, it only remains to prove the following lemma.

\begin{lemma}
The graph $G'$ allows a $(c+1,c)$-coloring if and only if $G$ allows a $(1,c)$-coloring.
\end{lemma}

\begin{proof}
Suppose $\varphi$ is a $(c+1,c)$-coloring of $G'$. It follows from Lemma~\ref{l:Hink=l+1abcolor} that every vertex $u_1\in V(G_1)$ has $c$ neighbors of its own color in the corresponding bridge $H_u$, and thus the restriction of $\varphi$ to $G_1$ is a $(1,c)$-coloring of $G_1$ (which is isomorphic to $G$).

If $G$ allows a $(1,c)$-coloring, use it on both $G_1$ and $G_2$ and color the bridges so that for every $u\in V(G)$, the $r_i$ and $t_i$ vertices of $H_u$ get the same color as $u$ and the vertices $s_i$ and $w_i$ of $H_u$ get the opposite color. This is a $(c+1,c)$-coloring of $G'$.
\end{proof}

This concludes the proof of Proposition~\ref{p:k=l+1coloring}.
\end{proof}

\begin{proposition}\label{t:aa-coloring-np}
For every $b\ge 2$, the {\sc $(b,b)$-Coloring} problem is NP-complete even for simple bipartite graphs. 
\end{proposition}

\begin{proof}
We will reduce from the following problem.
\computationproblem
{\sc ($k$-in-$2k$)-SAT$_q$}
{A formula $\phi$ with clauses $C_1, \ldots, C_{m}$ in CNF without 
negations, each $C_i$ is a disjunction of exactly $2k$ distinct literals and every variable occurs exactly $q$ times.}
{Does there exist a satisfying assignment of $\phi$ such that exactly $k$ literals are true in each $C_i$?}

The problem {\sc ($k$-in-$2k$)-SAT$_q$} was proved to be NP-complete by Kra\-to\-chvíl~\cite{kratochvil2003complexity} for every $k \geq 2, q \geq 3$.

Let $\phi$ 
 be an instance of {\sc ($b$-in-$2b$)-SAT$_q$}, $b \geq 2$, with each variable occurring $q=b+1$ times. Let $C_1, \ldots, C_{m}$ be the clauses of $\phi$.

Our clause gadget is a complete bipartite graph $K_{b, 2b}$. The vertices in the bigger part correspond to variables. 
More formally, for every variable $x$ occurring in  a clause $C_i$, the clause gadget has a vertex $y^i_x$ in its bigger part. To 
make
 sure that each variable has an even number of occurrences, we will duplicate each clause gadget 
and we will refer to these copies as the \emph{left} and \emph{right} ones, with their $y$ vertices being denoted by $y^{i,l}_x$ and $y^{i,r}_x$, respectively.  

For each variable $x$, we will construct a variable gadget $V_{b}$ in the following way. Take complete bipartite graph $K_{2b+1,2b+1}$ and denote its vertices in one part as $u_1, \ldots, u_{2b+1}$ and in the other part as $v_1, \ldots, v_{2b+1}$. Remove the edges $u_iv_i$ for each $1 \leq i \leq 2b+1$ and the edges $u_iv_{i+b}$ for each $2 \leq i \leq b+1$. Take two copies $K_1, K_2$ of the resulting graph and add a new vertex connected to
$v_{b+2}, 
\ldots, v_{2b+1}$ in $K_1$ and $u_2, \ldots, u_{b+1}$ in $K_2$. 

Add a new vertex connected to $u_2, \ldots, u_{b+1}$ in $K_1$
(this vertex will be called the \emph{left} vertex)
and add a new vertex connected to 
$v_{b+2},
\ldots, v_{2b+1}$ in $K_2$ (called
the \emph{right} one). Take $b+1$ disjoint copies of this graph and add $2b+2$ new vertices $x_1,\ldots,x_{2b+2}$ which shall correspond to the occurrences of the variable $x$. We shall call $x_1,\ldots,x_{b+1}$ 
the
\emph{left occurrences} of $x$ and $x_{b+2},\ldots,x_{2b+2}$ 
the \emph{right occurrences} of $x$. 

Now we shall insert edges between the left occurrences of $x$  and the left
vertices so that they induce a $b$-regular bipartite graph with one part
being $x_1,\ldots,x_{b+1}$ and the second one being the left vertices. An
analogous construction will be done with $x_{b+2},\ldots,x_{2b+2}$ and the
right vertices. See Figure~\ref{fig:vargadget} for an example.

\begin{figure}
\centering
\includegraphics[scale=0.9]{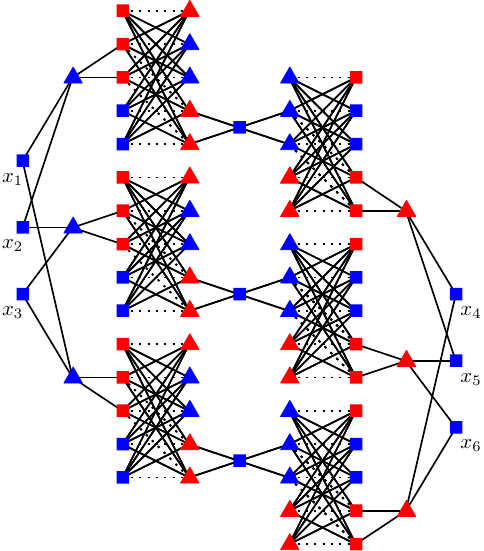}
\caption{A variable gadget $V_{2}$ for variable $x$ with a
$(b,b)$-coloring 
corresponding to valuation $\pi(x)$ = true.}
\label{fig:vargadget}
\end{figure}

To complete the construction, in the left copy of each clause gadget, we
identify each vertex of the part of the size $2b$ with the respective left
occurrences of the variable $x$ and in the right copy of each clause gadget,
we identify each vertex of the part of the size $2b$ with the respective
right occurrences of the variable $x$. Formally, if $C_i$ is the $j$-th
clause containing the variable $x$, we identify $y^{i,l}_x$ with $x_j$ and
$y^{i,r}_x$ with $x_{b+1+j}$. The resulting graph shall be called $G$. 

We claim that the formula $\phi$ is satisfiable if and only if $G$ has a $
(b,b)$-coloring.

First suppose that $\phi$ is satisfiable and take some satisfying assignment
$\pi$. We will construct a suitable coloring in the following way. For a
variable $x$, if $\pi(x)= $ true, then color $x_1,\ldots,x_{2b+2}$ by blue
color and otherwise, color all $x_1,\ldots,x_{2b+2}$ by red color. Color all
vertices in the smaller parts of the left copies of clause gadgets by red
color and all vertices in the smaller parts of the right copies of clause
gadgets by blue color. In the variable gadgets, vertices of one class of
bi-partition will be colored the same regardless the value of the
corresponding variable while the coloring of the the other class of
bi-partition will depend on its value. The left vertices
(connecting $x_1,\ldots,x_{b+1}$ to $K_1$) will be all colored blue, the
right vertices (connecting $x_{b+2},\ldots,x_{2b+2}$ to $K_2$) will be all
colored red. The $v_i$'s of $K_1$'s will always be colored so that $v_1$ and
$v_{b+2},\ldots,v_{2b+1}$ are red and $v_{2},\ldots,v_{b+1}$ are blue, the
$u_i$'s of $K_2$'s will always be colored so that  $u_1,\ldots,u_{b+1}$ are
blue and $u_{b+2},\ldots,u_{2b+1}$ are red. In the other class of
bi-partition, if $\pi(x)=$ true, then on top of all the occurrences
$x_1,\ldots, x_{2b+2}$, also all the "middle" vertices connecting $K_1$'s to
$K_2$'s, the vertices $u_{b+2},\ldots,u_{2b+1}$ in $K_1$'s and the vertices
$v_2,\ldots,v_{b+1}$ in $K_2$'s will be colored blue, while the vertices
$u_1,\ldots,u_{b+1}$ of $K_1$'s and the vertices $v_1,v_{b+2},\ldots,v_
{2b+1}$ in $K_2$'s will be colored red. If $\pi(x)=$ false, the colors of the
vertices in this class of bipartition will be swapped. See an example in
the Figure~\ref{fig:vargadget} for a variable evaluated to true. Since in
every clause, there are exactly $b$ variables set to true, all vertices
in the smaller parts of clause gadgets have exactly $b$ red and exactly $b$
blue neighbors. It can be shown by a detailed case analysis that the same
holds for all vertices, and so this is a $(b,b)$-coloring of $G$.

Suppose that $G$ has a $(b,b)$-coloring, and fix one such coloring. For a
variable $x$, we set $x$ to be true if all $x_1,\ldots,x_{2b+2}$ are colored
by blue color and we set it to be false if all $x_1,\ldots,x_{2b+2}$ are
colored by red color. We need to prove that such assignment always exists and
that it is a satisfying assignment. 

First we prove that in every $(b,b)$-coloring either all of $x_1,\ldots,x_
{2b+2}$ are colored blue or all of $x_1,\ldots,x_{2b+2}$ are colored red.
Recall the subgraph $K_1$ of a variable gadget with vertices $u_1, \ldots, u_
{2b+1}$ in one part and $v_1, \ldots, v_{2b+1}$ in the other part.

We claim that in every $(b,b)$-coloring of $V_b$ restricted to some copy of
$K_1$ and its two adjacent vertices, the vertices $u_2, \ldots, u_{b+1}$ are
either all red or all blue. Suppose for a contradiction that in some $
(b,b)$-coloring there exist indices $i,j \in \{2, \ldots, b+1\}$ such that
$u_i$ is colored by red and $u_j$ is colored by blue. Since $v_1$ is adjacent
to all $u_2, \ldots, u_{2b+1}$, exactly $b$ of them are colored red and
exactly $b$ of them are colored blue. Since $v_i$ is not adjacent to $u_i$,
we need to color $u_1$ by red. However, since $v_j$ is not adjacent to $u_j$,
we have to color $u_1$ by blue, a contradiction. 

Suppose without loss of generality that all $u_2, \ldots, u_{b+1}$ are blue.
As argued above, all $u_{b+2},\ldots,u_{2b+1}$ are then red. All of them are
neighbors of $v_2$, and hence $u_1$ is blue. Let $w$ be the vertex outside of
$K_1$ adjacent to $ v_{b+2},
 \ldots, v_{2b+1}$ in $K_1$. Since $v_{2b+1}$ has only $b-1$ red neighbors in
$K_1$, $w$ must be red. Similar arguments apply to $K_2$. Thus,
$u_2, \ldots, u_{b+1}$ in $K_1$ and $v_{b+1}, \ldots, v_{2b+1}$ in $K_2$
always have the same color. Then all $b$ occurrences of the variable
adjacent to the left vertex of $K_1$ and all $b$ occurrences adjacent to
the the right vertex of $K_2$ get the same color. Since $b\ge 2$, it
follows from the construction between the occurrences and variable gadgets
that all occurrences of the variable have the same color.

It remains to be proven that this is a satisfying assignment. Since the
vertices of the smaller parts of clause gadgets have degree $2b$, exactly $b$
vertices of the bigger part of each clause are colored by red and exactly $b$
vertices of the bigger part of each clause are colored by blue. Thus, exactly
$b$ variables in each clause are set to be true. This concludes the proof.
\end{proof}

\section{Conclusion}\label{s:conclusion}

The main goal of this paper was to initiate the study of the computational complexity of covering graphs with semi-edges. 
We have exhibited a new level of difficulty that semi-edges bring to coverings by showing a connection to edge-colorings. 
We have presented a complete classification of the computational complexity of covering graphs with at most two
vertices, which is already a quite nontrivial task. 
In the case of one-vertex target graphs, the problem becomes polynomial-time solvable if the input graph is bipartite, while in the case of two-vertex target graphs, bipartiteness of the input graphs does not help. This provides a strengthening of known results on covering two-vertex graphs without semi-edges.

It is worth noting that the classification in~\cite{n:KPT97a} concerns a more general class of {\em colored mixed} (multi)graphs, i.e., graphs which may have both directed and undirected edges and whose edges come with assigned colors which must be preserved by the covering projections. It turns out that covering a two-vertex (multi)graph is NP-hard if and only if it is NP-hard for at least one of its maximal monochromatic subgraphs. It is shown in~\cite{BokFJKS-DAM24} that the same holds true when semi-edges are allowed (note that all semi-edges must be undirected only).

We end up with an intriguing open problem.

\begin{problem}
Do there exist graphs $H_1$ and $H_2$, both without semi-edges, such that $H_1$ covers $H_2$, and such that the $H_1$-{\sc Cover} is polynomial-time solvable and $H_2$-{\sc Cover} is NP-complete?	
\end{problem}

If semi-edges are allowed, then $H_1=W(0,0,k,0,0)$ and $H_2=F(k,0)$, with $k\ge 3$, is such a pair.

\section*{Acknowledgments}

The authors are grateful to the two anonymous reviewers for their constructive feedback and suggestions.

\bibliographystyle{splncs04}
\bibliography{bib/knizky,bib/nakryti,bib/sborniky,0-main}

\end{document}